\documentclass[a4paper,fleqn]{cas-dc}

\usepackage[numbers]{natbib}
\usepackage{algorithm}
\usepackage{algpseudocode}
\usepackage{amsmath,mathtools,amssymb}

\newtheorem{theorem}{Theorem}
\newtheorem{lemma}{Lemma}

\newtheorem{proposition}{Proposition}

\newtheorem{ass}{Assumption}
\newtheorem{definition}{Definition}

\newcommand{\R}{\mathbb{R}}

\newcommand\norm[1]{\ensuremath{\lVert#1\rVert}}

\DeclareMathOperator\diag{diag}

\def\tsc#1{\csdef{#1}{\textsc{\lowercase{#1}}\xspace}}
\tsc{WGM}
\tsc{QE}

\begin{document}
\let\WriteBookmarks\relax
\def\floatpagepagefraction{1}
\def\textpagefraction{.001}

\shorttitle{}    


\title [mode = title]{Physics-informed structured learning of a class of recurrent neural networks with guaranteed properties}  



%

\author[1,2]{Daniele Ravasio}
\ead{daniele.ravasio@polimi.it}
\cormark[1]
\author[1]{Claudia Sbardi}
\ead{claudia.sbardi@mail.polimi.it}
\author[1]{Yushan Li}
\ead{yushan.li@polimi.it}
\author[1]{Marcello Farina}
\ead{marcello.farina@polimi.it}
\author[2]{Andrea Ballarino}
\ead{andrea.ballarino@stiima.cnr.it}





\affiliation[1]{organization={Dipartimento di Elettronica, Informazione e Bioingegneria (DEIB), Politecnico di Milano},
            addressline={Via Ponzio 34/5}, 
            city={Milan},
            postcode={20133}, 
            country={Italy}}
\affiliation[2]{organization={Istituto di Sistemi e Tecnologie Industriali Intelligenti per il Manifatturiero Avanzato (STIIMA) - Consiglio Nazionale delle Ricerche (CNR)},
            addressline={Via A. Corti 12}, 
            city={Milan},
            postcode={20133}, 
            country={Italy}}






\makeatletter
\def\printorcid{}
\makeatother



\begin{abstract}
This paper proposes a physics-informed learning framework for a class of recurrent neural networks tailored to networked systems. The approach aims to learn control-oriented models that preserve the structural and stability properties of the plant. The learning algorithm is formulated as a convex optimisation problem, allowing the inclusion of linear matrix inequality constraints to enforce desired plant features. Furthermore, when the plant exhibits structural modularity, the resulting optimisation problem can be parallelised, requiring communication only among neighbouring subplants. Simulation results show the effectiveness of the proposed approach.
\end{abstract}





\maketitle

\section{Introduction}
\subsection{Motivation}
The modelling and control of networked systems composed of multiple interacting subsystems is a research area that has attracted increasing attention due to the challenges posed by their complexity, dimensionality, and distributed structure~\cite{maestre2014distributed,tang2018network}. 
Examples include power networks, manufacturing processes, and transportation systems. Within this setting, centralised control strategies are often inadequate due to scalability issues, communication constraints, and privacy concerns. A common approach to address these challenges is to exploit the structural modularity of the plant to decompose the overall control problem into smaller, weakly coupled subproblems, each managed by a local controller, leading to decentralised or distributed control architectures~\cite{scattolini2009architectures}. However, the design of decentralised or distributed control schemes relies on the availability of accurate plant models that exhibit a structure consistent with the intended decomposed control system architecture~\cite{tang2018network,farina2018distributed}. 

The definition of first-principles models for networked systems often becomes impractical due to their complexity, large number of parameters, and the physical interactions among subsystems.
In this context, recurrent neural networks (RNNs)~\cite{lecun2015deep} have emerged as powerful tools for modelling complex plants, owing to their ability to capture long-term and nonlinear temporal dependencies directly from data. Nevertheless, purely black-box models may fail to capture the underlying comprehensive physical/structural properties of the real plant. This lack of physical consistency may result in unreliable or non-interpretable models, which might be inadequate for control design.\\ 
This limitation has motivated recent research on physics-informed learning strategies~\cite{bonassi2022recurrent,hao2022physics,10155901} aimed at incorporating available prior knowledge of the plant directly into the model training process. The available prior knowledge may include not only first-principles equations but also qualitative and structural properties. In particular, structural information can be exploited to constrain the model architecture to reflect the plant modularity. Similarly, when the plant exhibits stability properties, such as incremental input-to-state stability ($\delta$ISS,~\cite{bayer2013discrete}), it is desirable to enforce them to the RNN model. Enforcing stability properties at the learning stage is relevant from two complementary perspectives~\cite{bonassi2022recurrent}. First, from a physical perspective, it aims to ensure consistency with prior knowledge of the plant qualitative behaviour, thereby improving the reliability of the obtained model. Second, from a control-oriented perspective, it provides a theoretical tool that can be directly leveraged during control design. This idea has been explored in several works for the design of control schemes with stability and performance guarantees (see, e.g.,~\cite{bonassi2024nonlinear,schimperna2024robust}). $\delta$ISS, in particular, is a robust stability property that guarantees the existence of robust positively invariant sets~\cite{bayer2013discrete}, which are essential ingredients in the design of robust control algorithms~\cite{ravasio2024lmi,ravasio2026recurrent}.

A further challenge is that standard RNN training techniques are generally computationally demanding and centralised~\cite{keuper2016distributed}, which may limit their applicability to networked plants, particularly when the number of subplants increases. Moreover, data from individual subplants may not be fully accessible to a central computational unit, e.g., due to privacy constraints or the geographical distribution of the subplants. These challenges motivate the development of computationally efficient and scalable learning strategies that limit the information exchange across the network. 
%
%
\subsection{Statement of the problem}
In this work we assume that the plant $\mathcal{P}$ under analysis is endowed with manipulable inputs, collected in the input vector $u\in\mathbb{R}^m$, and measurable outputs, collected in the output vector $y\in\mathbb{R}^p$.
The aim of this work is to develop a framework for learning a physics-informed recurrent neural network model of the plant $\mathcal{P}$, in which selected physical features of $\mathcal{P}$ are directly embedded in the learning process.
Concerning the available data, we make the following assumption. 
\begin{ass}\label{ass:dataset}
An informative dataset of input-output data previously collected from the plant is available. The data consists of an applied input sequence $\mathcal U_\mathrm d=\{u_\mathrm d(k)\}_{k=1}^{N_\mathrm d}$, where the scalar $N_\mathrm d$ represents the sequence length, and a measured output sequence $\mathcal Y_\mathrm d=\{y_\mathrm d(k)\}_{k=1}^{N_\mathrm d}$. 
\end{ass}
The resulting model is intended to provide an accurate representation suitable for control design. The focus is on two key aspects: \emph{(i)} the construction of a modular model inspired by the plant structure, which can be directly used in the synthesis of a distributed or a decentralised control scheme, and \emph{(ii)} the incorporation of stability guarantees, so that the resulting model inherits the stability properties of $\mathcal{P}$.
\subsubsection{Imposing the modular plant structure}
\label{subsec:modular_system}
Many 
complex plants are characterised by a structural modularity which can be unveiled by physical inspection or through data-driven approaches (see, e.g., \cite{materassi2012problem}).
In particular, we can often define a number $n_\mathrm s$ of subplants $\mathcal{P}_i$, where $i\in\mathcal{I}=\{1,\dots,n_\mathrm s\}$, each characterised by a local manipulable input vector $u_i\in\R^{m_i}$ and a vector $y_i\in\R^{p_i}$ of local measurable outputs. In this work we assume that $\sum_{i=1}^{n_\mathrm s} m_i=m$ and $\sum_{i=1}^{n_\mathrm s} p_i=p$, in such a way that $u$ and $y$ are partitioned in a non-overlapping fashion, i.e., 
$u=[
         u_1^\top,\,\dots,\,u_{n_\mathrm s}^\top
     ]^\top$ and
     $y=[
         y_1^\top,\,\dots,\,y_{n_\mathrm s}^\top
     ]^\top.$\\
We assume here that physical interconnections among subplants can also be defined, see, e.g.,  Figure \ref{fig:modular_S}. 
\begin{figure}
     \centering
    \includegraphics[width=0.7\columnwidth]{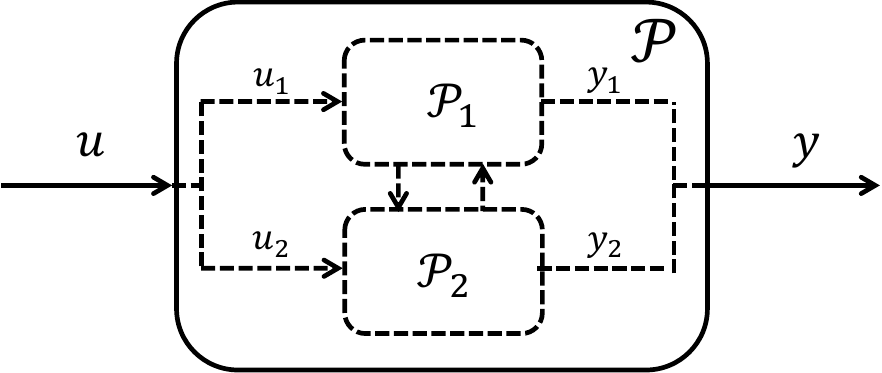}
    \caption{Modular plant structure for a simple case where $n_\mathrm s=2$.}
      \label{fig:modular_S}
\end{figure}
In general, we say that the subplant $\mathcal{P}_j$ has a direct influence on (or, in a graph-theoretical terminology, is neighbor of) subplant $\mathcal{P}_i$ (with $i\neq j$) if there exists an interconnection vector $\nu^{\mathcal{P}}_{ij}$ of variables of $\mathcal{P}_j$ which is an input for $\mathcal{P}_i$. Vector $\nu^{\mathcal{P}}_{ij}$ may include entries of $u_j$ and of $x_j$, the latter defining the internal state of $\mathcal{P}_j$ (some of whose elements may be measurable, i.e., included in $y_j$). This induces the definition of the graph $\mathcal{G}^{\mathcal{P}}=\{\mathcal{I},\mathcal{V}^{\mathcal{P}}\}$, whose edges correspond to the subplants $\mathcal{P}_i$ for all $i\in\mathcal{I}$, and  the pair $(i,j)$ is a vertex (i.e., $(i,j)\in\mathcal{V}^{\mathcal{P}}$) if and only if $\nu^{\mathcal{P}}_{ij}\neq 0$.\\
More specifically, for all $i\in\mathcal I$ we can define the following neighboring sets: $\mathcal{N}^{\mathcal{P}}_{u,i}\coloneq\{j\in\mathcal I\setminus\{i\}:\nu^{\mathcal{P}}_{ij}\neq 0$ includes entries of $u_j\}$ and
$\mathcal{N}^{\mathcal{P}}_{x,i}\coloneq\{j\in\mathcal I\setminus\{i\}:\nu^{\mathcal{P}}_{ij}\neq 0$ includes entries of $x_j\}$.\\
%
In Section \ref{sec:learning structured models} we propose a methodology to identify $n_{\mathrm s}$ submodels $\mathcal{M}_i$, with $i=1,\dots,n_\mathrm s$, each endowed with an input/output pair $(u_i,y_i)$, and interconnected with each other through an interconnection topology induced by the modular plant structure defined above.\\ 
Note that, even in model-based contexts, the model decomposition problem is crucial and critical. In fact, as analysed, e.g., in~\cite{farina2018distributed}, while overlapping decompositions lead to higher-order submodels $\mathcal{M}_i$ with dense (e.g., all-to-all) interconnection structures, non-overlapping decompositions typically lead to submodels $\mathcal{M}_i$ with reduced orders and minimal interconnection structures (which may, unless we perform further simplifications, correspond to the plant real interconnection structures). In view of the better scalability of the latter decomposition, in Section \ref{sec:learning structured models} we propose a data-based counterpart of non-overlapping decompositions where we learn submodels, partially coupled through inputs and/or internal states, reflecting the plant's interconnection graph.
%
%
\subsubsection{Imposing stability guarantees}
In this work, in view of the fact that we cast the learning problem as a convex optimisation one, we can enforce the stability properties exhibited by the plant through suitable constraints expressed in the form of linear matrix inequalities (LMIs)~\cite{boyd1994linear}. In Section~\ref{sec:diss_learning} we propose an LMI-based sufficient condition on the weights of the considered model class that guarantees its $\delta$ISS and exploit this condition to enforce this property, during the learning phase, on the model. Notably, when the plant is characterised by a modular structure, the proposed methodology can be leveraged to impart the $\delta$ISS property both to the local submodels $\mathcal M_i$, $i = 1,\dots,n_s$, and to the overall structured model resulting from their interconnection.\\
Note that the methodology proposed in Section~\ref{sec:diss_learning} is not limited to the $\delta$ISS and can be extended to enforce alternative properties, provided that suitable conditions on the weights of the considered RNN model class can be expressed in terms of LMIs.
%
%
%
\subsection{State of the art}
The inclusion of physical knowledge in learning strategies for RNN models has received considerable attention in the literature as an attempt to overcome some of the limitations of black-box modelling. 
Along this line, several approaches have been developed within standard gradient-based training frameworks~\cite{bonassi2022recurrent,hao2022physics}. 

One of the most widely adopted approaches to incorporate physical constraints consists in augmenting the training objective with suitably designed regularisation terms. Physics-guided loss functions can be designed to enforce physical laws such as conservation principles or stability inequalities. For instance, this strategy has been employed to enforce ISS and $\delta$ISS conditions in various RNN architectures~\cite{bonassi2022recurrent}.\\
However, this approach suffers from several limitations~\cite{krishnapriyan2021characterizing}. First of all, the conditions to be imposed on the model parameters to impart stability are commonly conservative, often compromising the quality of the so-obtained learned models. Also, since the these conditions are not enforced as strict constraints, there are no guarantees that the resulting model will satisfy the desired physical properties, necessitating a posteriori verification. Additionally, if not designed appropriately, these regularisation terms may lead to optimisation issues.

An alternative to physics-guided loss functions is to enforce a-priori physical knowledge through hard constraints~\cite{hao2022physics} or by using ad-hoc model structures~\cite{10155901}. However, the use of hard constraints, for general RNNs, may lead to complex and non-convex optimisation programs~\cite{pan2021can}, possibly increasing the computational burden of the learning problem. Also, in case specific model structures are imposed as described in~\cite{10155901}, the resulting models may not be sufficiently flexible to include multiple physics-informed features.\\
A notable structure is the class of recurrent equilibrium networks (RENs) introduced in~\cite{revay2023recurrent}, allowing us to enforce contractivity and robust stability through LMI constraints. However, the approach in~\cite{revay2023recurrent} may become difficult to apply when one seeks to impose, at the same time, structural constraints on the model matrices, for instance, to reflect the interconnection topology of a networked plant. 
Structural information can indeed be reflected in the model by adopting a sparsely connected RNN that is consistent with the plant topology, as done in \cite{wu2020process} for enforcing structural properties in chemical flowsheet models and in \cite{bonassi2022recurrent} for learning a structured RNN model of a chemical process.

However, the approaches adopted so far employ gradient-based training strategies (such as the widely adopted \textit{backpropagation through time}), which are generally inherently non-convex and prone to the presence of local minima. Moreover, especially when physically-inspired constraints and specific model structures are imposed, they suffer from significant scalability limitations and are difficult to parallelise~\cite{lecun2015deep,keuper2016distributed}.\\ 
An interesting alternative that avoids gradient-based optimisation is provided by echo state networks (ESNs) \cite{jaeger2001echo}. In ESNs, the matrices of the so-called dynamic reservoir are fixed and generated before training. Consequently, the identification problem is reduced to estimating the parameters of the output relation and is formulated as a least-squares program. This method is computationally very efficient; however, in the presence of measurement noise, the resulting estimator can become biased due to the possible correlation of the output noise with the state predictor, potentially degrading the resulting model performance. To overcome this limitation, the set-membership identification framework proposed in \cite{d2025data} formulates the learning problem by exploiting deterministic bounds on the measurement noise. Although computationally more demanding than least squares, this approach enables the computation of a bound on the modelling uncertainty.
\subsection{Paper objectives and structure}
Motivated by the considerations made in the previous sections, in this work we generalise the least-squares \cite{jaeger2001echo} and set-membership \cite{d2025data} learning approaches developed for ESNs to a broader class of RNNs having the structure of REN models. In particular, we show that both learning problems can be reformulated as convex optimisation problems involving LMIs, thus preserving convexity. Building upon this formulation, we develop a framework for the identification of control-oriented models that preserve structural and stability features of the plant. Importantly, when the plant exhibits structural modularity, the resulting learning problem can be parallelised, requiring communication only among neighbouring subplants, and resulting scalable and numerically well suited to tackle networked plants.\\
%
The paper is structured as follows. In Section \ref{sec:model} we first define the considered model class, its features, as well as its well-posedness property. Then, in Section \ref{sec:training_alg} we define the main approaches used for well-posed model training; these approaches are exploited in Sections \ref{sec:learning structured models} and \ref{sec:diss_learning} in case we need to identify models with a plant-inspired modular structure and with embedded guaranteed stability properties, respectively. For clarity, Table~\ref{tab:rnn_algorithms_summary} provides an overview of the proposed learning algorithms. Simulation tests are reported in Section \ref{sec:SIMULATIONS} and conclusions are drawn in Section \ref{sec:conclusions}. 
For clarity of exposition, all proofs are reported in the Appendix.
\begin{table*}[htbp]
\centering
\caption{Summary of the learning algorithms.}
\label{tab:rnn_algorithms_summary}
\smallskip
\begin{tabular}{p{2.3cm} p{5.0cm} p{8.7cm}}
\toprule
\textbf{Architecture} &
\textbf{Well-posedness} &
\textbf{Well-posedness and $\delta$ISS} \\
\midrule

\makecell[l]{\textbf{Unstructured}\\ \textit{Least-squares}}
&
\textbf{Algorithm~\ref{alg:LS_learning}}
&
\textbf{Algorithm~\ref{alg:LS_learning}} with
Step~2 replaced by
\textbf{Algorithm~\ref{alg:LS_identification_diss}}.
\\

\addlinespace

\makecell[l]{\textbf{Unstructured}\\ \textit{set-membership}}
&
\textbf{Algorithm~\ref{alg:set_membership_learning}}
&
\textbf{Algorithm~\ref{alg:set_membership_learning}} with
Step~4 replaced by
\textbf{Algorithm~\ref{alg:scenario_parameter_selection_diss}}.
\\

\addlinespace

\makecell[l]{\textbf{Structured}}
&
\textbf{Algorithm~\ref{alg:distributed_identification}},
using:
\begin{itemize}
\item \textbf{Algorithm~\ref{alg:LS_identification}} (least-squares), or
\item \textbf{Algorithms~\ref{alg:FPS_computation}} and
\textbf{\ref{alg:scenario_parameter_selection}}
(set-membership).
\end{itemize}
&
\textbf{Submodels $\delta$ISS:}

\textbf{Algorithm~\ref{alg:distributed_identification}} with
Step~8 replaced by
\textbf{Algorithm~\ref{alg:LS_identification_diss}} (least-squares)
and Step~17 replaced by
\textbf{Algorithm~\ref{alg:scenario_parameter_selection_diss}} (set-membership).

\vspace{0.5em}

\textbf{Overall model $\delta$ISS:}

\textbf{Algorithm~\ref{alg:distributed_identification}} with
Step~8 replaced by
problem~\eqref{opt:LS_identification_structured_diss} (least-squares) and
Steps~16--18 replaced by
problem~\eqref{opt:scenario_sampling_structured_diss} (set-membership).
\\

\bottomrule
\end{tabular}
\end{table*}
\subsection{Notation and preliminaries} 
Given a vector $v \in \mathbb{R}^n$, $v^{(i)}$ denotes its $i$-th entry. Given a matrix $M \in \mathbb{R}^{n \times n}$, $M^{(i)}$ denotes its $i$-th row, and $M^\top$ its traspose. Let $\mathbb{Z}_{+}$ denote the set of positive integers (excluding zero), $\mathbb{S}_+^n$ the set of real symmetric positive definite matrices, and $\mathbb{D}_+^n$ the set of real diagonal positive definite matrices. We denote the sequence $(u(0), \dots, u(N))$ by \(\{u(k)\}_{k=0}^N\). The matrix $I_n$ denotes the $n \times n$ identity matrix. Given $n$ matrices $M_{1},\dots,M_{n}$, we denote by $\diag(M_{1},\dots,M_{n})$ the block-diagonal matrix with $M_{1},\dots,M_{n}$ on its main diagonal blocks. Given $\epsilon\in(0,1)$, we denote as $\log_{1-\epsilon}(\cdot)$ the standard logarithm with a base equal to $1-\epsilon$.
Given an index set $\mathcal I=\{1,\dots,n\}\subseteq\mathbb Z_{+}$, its cardinality is denoted by $|\mathcal I|$. Assuming that all matrices $M_i$ have the same number of rows, we define $[M_i]_{i\in\mathcal I}\coloneq[M_{i_1}\,\dots\, M_{i_n}]$. Consider two index sets $\mathcal I, \mathcal J \subseteq \{1,\dots,n\}$ and a matrix $M$, which is partitioned into $n\times n$ non-overlapping blocks $M_{i,j}$. We denote this block partition by $M=[M_{i,j}]_{i\in\mathcal I, j\in\mathcal J}$.  
Also, we denote by $(M)_{\mathcal I,\mathcal J}$ the matrix obtained by keeping only the row-blocks of $M$ indexed by $\mathcal I$ and the column-blocks indexed by $\mathcal J$.
Given a vector $v$ partitioned into $m$ non-overlapping subvectors $v_i$, $i=1,\dots,m$, we denote by $(v)_{\mathcal I}$ the vector obtained by keeping only the subvectors $v_i$ such that $i \in \mathcal I$. Given a vector $v\in\R^m$ and a set $\tilde{\mathcal V}\subseteq\R^{m}$, we  define the function $\text{dist}(v,\tilde{\mathcal V})=\min_{\tilde v\in\tilde{\mathcal V}}\norm{v-\tilde v}^2$
A continuous function $\alpha:\R_{\geq0}\to\R_{\geq0}$ is a class $\mathcal K$-function if $\alpha(s)>0$ for all $s>0$, it is strictly increasing and $\alpha(0)=0$. Also, a continuous function $\alpha:\R_{\geq0}\to\R_{\geq0}$ is a class $\mathcal K_\infty$-function if it is a class $\mathcal K$-function and $\alpha(s)\to\infty$ as $s\to\infty$. Finally, a continuous function $\beta:\R_{\geq0}\times\mathbb Z_{\geq0}\to\R_{\geq0}$ is a class $\mathcal{KL}$-function if $\beta(s,k)$ is a class $\mathcal K$-function with respect to $s$ for all $k$, it is strictly decreasing in $k$ for all $s>0$, and $\beta(s,k)\to0$ as $k\to\infty$ for all $s>0$. 
\smallskip\\
Consider a general nonlinear discrete-time system described by,
\begin{align}
    x(k+1)&=f(x(k),u(k))\label{eq:general_sys}
\end{align}
where $k\in\mathbb Z_{\geq 0}$ is the discrete-time index, $x\in\R^n$ is the state vector, $u\in\R^m$ is the input vector, and $f:\R^n\times\R^m\to\R^n$.\\
We introduce the following definitions \cite{revay2023recurrent,bayer2013discrete}.
\begin{definition}\label{def:contractivity}
System~\eqref{eq:general_sys} is said to be contracting with rate $\alpha\in(0, 1)$ if, for any two initial conditions $x_a(0)$, $x_b(0) \in \mathbb{R}^{n}$, given the same sequence $\{u(k)\}_{k=0}^N$, where $N\in\mathbb Z_+$, the state sequences $x_a(k)$, $x_b(k)$ satisfy,
    \begin{equation*}
        \norm{x_a(k)-x_b(k)} \leq   \rho\alpha^k\norm{x_a(0)-x_b(0)},
    \end{equation*}
     for some $\rho\in\R_+$.
\end{definition}
\begin{definition}\label{def:delta_ISS}
    System~\eqref{eq:general_sys} is said to be $\delta$ISS with respect to $\mathcal X\subseteq\R^n$ and $\mathcal U\subseteq\R^m$ if $\mathcal X$ is robust positively invariant for~\eqref{eq:general_sys} and there exist functions $\beta\in\mathcal{KL}$ and $\gamma\in\mathcal{K_\infty}$ such that for any $k\in\mathbb{Z}_{\geq0}$, any pair of initial states $x_a(0),x_b(0)\in\mathcal X$, and any pair of input $u_a,u_b\in\mathcal U$, it holds that
    \begin{multline*}
    \norm{x_a(k)-x_b(k)}{\leq}\beta(\norm{x_a(0)-x_b(0)},k)\\+\gamma(\max_{h\geq 0}\norm{u_a(h)-u_b(h)}).
    \end{multline*}
\end{definition}
\begin{definition}
A function $V(x_a,x_b)$ is a dissipation-form $\delta$ISS Lyapunov function for system \eqref{eq:general_sys} if $\mathcal X$ is robust positively invariant for~\eqref{eq:general_sys} and there exist $\mathcal{K_\infty}$-functions $\alpha_1$, $\alpha_2$ and $\alpha_3$, and a $\mathcal{K}$-function $\alpha_4$ such that, for any pair of states $x_a(k)\in\mathcal{X}$ and $x_b(k)\in\mathcal{X}$, and any pair of inputs $u_a(k)\in\mathcal{U}$ and $u_b(k)\in\mathcal{U}$, it holds that
    \begin{equation*}
       \alpha_1(\norm{x_a(k)-x_b(k)})\!\leq\! V(x_a(k),x_b(k))\!\leq\!\alpha_2(\norm{x_a(k)-x_b(k)}),
    \end{equation*}
    \begin{multline*}
    V(x_a(k+1),x_b(k+1))-V(x_a(k),x_b(k))\\\leq -\alpha_3(\norm{x_a(k)-x_b(k)})+\alpha_4(\norm{u_a(k)-u_b(k)}).
    \end{multline*}
\end{definition}
A sufficient condition such that system~\eqref{eq:general_sys} is $\delta$ISS is stated in the following theorem~\cite{bayer2013discrete}. 
\begin{theorem}\label{th:local_dISS}
   If system~\eqref{eq:general_sys} admits a dissipation-form
    $\delta$ISS Lyapunov function over the sets $\mathcal{X}$ and $\mathcal{U}$, then it is $\delta$ISS with respect to such sets, in the sense of Definition~\ref{def:delta_ISS}.
\end{theorem}
%
%
%
%
\section{The selected recurrent neural network class}
\label{sec:model}
\subsection{The recurrent neural network model}
The RNN considered in this paper is a deep (i.e., multi-layer) architecture comprising $n$ neurons in the reservoir, whose states are collected in the vector $x\in\R^n$. The RNN takes an input $u\in\R^m$ and produces an output $y\in\R^p$. In order to convexify the learning process, we adopt a training approach inspired by the ESN training algorithm proposed in~\cite{jaeger2001echo}. Accordingly, the RNN model is described as
\begin{subequations}\label{eq:RNN_model}
\begin{align}
&x(k+1)\! =\! A_xx(k) + B_uu(k) + B_s^\mathrm{o}s(k) + B_yy(k) \label{eq:RNN_model_x}\\
&s(k) = \sigma(\tilde A_xx(k) + \tilde B_uu(k) + \tilde B_s^\mathrm{o}s(k) + \tilde B_yy(k)) \label{eq:RNN_model_s}\\
&y(k) = Cx(k) + Du(k) + D_s s(k) \label{eq:RNN_model_y}
\end{align}
\end{subequations}
Matrices $A_x\in\R^{n\times n}$, $B_u\in\R^{n\times m}$, $B_s^\mathrm{o}\in\R^{n\times \nu}$, $B_y\in\R^{n\times p}$, $\tilde A_x\in\R^{\nu\times n}$, $\tilde B_u\in\R^{\nu\times m}$, $\tilde B_s^\mathrm{o}\in\R^{\nu\times \nu}$, and $\tilde B_y\in\R^{\nu\times p}$ are treated as hyperparameters, which are selected before training. 
On the other hand, matrices $C\in\R^{n\times p}$, $D\in\R^{p\times m}$, and $D_s\in\R^{p\times \nu}$ are free trainable parameters. Moreover,  $\sigma(\cdot):\R^\nu\to\R^\nu$ is a decentralised vector of sigmoidal activation  functions applied element-wise, i.e., $\sigma(v)=\begin{bmatrix}\sigma^{(1)}(v^{(1)})& \dots& \sigma^{(\nu)}(v^{(\nu)})\end{bmatrix}^\top,$ where $\sigma^{(i)}(\cdot)$, for $i=1,\dots,\nu$, are hyperparameters that must be selected so as to fulfil the following assumption.
\begin{ass}\label{ass:sigmoid_function}
    Each component $\sigma^{(i)}:\R\rightarrow\R$, $i=1,\dots,\nu$, is a sigmoid function, i.e., a bounded, twice continuously differentiable function with positive first derivative at each point and one and only one inflection point at the origin. Also, $\sigma^{(i)}(\cdot)$ is Lipschitz continuous with unitary Lipschitz constant and
such that $\sigma^{(i)}(0)=0$, $\frac{\partial\sigma^{(i)}(v^{(i)})}{\partial v^{(i)}}\big|_{v^{(i)}=0}=1$ and $\sigma^{(i)}(v^{(i)})\in[-1,1]$, $\forall v^{(i)}\in\R$.
\end{ass}
\subsection{Hyperparameters definition}\label{sec:hyperparameters}
Similarly to the ESN training algorithm~\cite{jaeger2001echo}, the hyperparameters $n$ (i.e., the RNN order), $\nu$ (i.e., the number of entries of $s$), $\sigma^{(i)}$ (the nonlinearity type), for $i=1,\dots,\nu$, and the matrices in \eqref{eq:RNN_model_x}-\eqref{eq:RNN_model_s} are user-generated before training. We collect these pre-selected quantities in the hyperparameter set
\(
\mathcal{H} \coloneqq \{ n, \nu, \sigma, A_x, B_u, B^o_s, B_y, \tilde{A}_x, \tilde{B}_u, \tilde{B}^o_s, \tilde{B}_y \}.
\)\\
In particular, matrices ($B_u$, $B_s^o$, $B_y$, $\tilde B_u$, $\tilde B_y$) can take random values. On the other hand, matrices ($A_x$, $\tilde A_x$, $\tilde B_s^\mathrm o$) must be defined according to the following proposition.
\begin{proposition}[\cite{revay2023recurrent}]\label{prop:contractivity_untrained}
    Under Assumption~\ref{ass:sigmoid_function}, if there exist $\bar\alpha\in(0,1)$ and matrices $P_\mathrm o\in\mathbb S_+^{n}$ and $\Lambda_\mathrm o\in\mathbb D_+^\nu$, such that
\begin{equation}
\resizebox{\columnwidth}{!}{$
\begin{bmatrix}
    \bar{\alpha}^2 P_\mathrm{o} & -\tilde{A}_x^\top \Lambda_\mathrm{o} \\
    -\Lambda_\mathrm{o} \tilde{A}_x & 2\Lambda_\mathrm{o} - \Lambda_\mathrm{o}\tilde{B}_s^{\mathrm{o}} - \tilde{B}_s^{\mathrm{o}\,\top}\Lambda_\mathrm{o}
\end{bmatrix}
 - \begin{bmatrix}
            A_x^\top \\ {B_s^{\mathrm{o}}}^\top
        \end{bmatrix} P_\mathrm o \begin{bmatrix}
            A_x & B_s^{\mathrm{o}}
        \end{bmatrix} \succ0,$}
        \label{eq:RENcontractivity_Lyapunov}
    \end{equation}
     then, model~\eqref{eq:RNN_model} is well-posed and contracting with rate $\alpha{<}\bar{\alpha}$.
\end{proposition}
Note that, by applying the Schur complement to~\eqref{eq:RENcontractivity_Lyapunov} and by substituting $Z_x=P_\mathrm oA_x$, $\tilde Z_x=\Lambda_\mathrm o \tilde A_x$, and $\tilde Z_s=\Lambda_\mathrm o \tilde B_s^\mathrm o$, condition~\eqref{eq:RENcontractivity_Lyapunov} is equivalent to
\begin{equation}\label{eq:contractivity_untrained_LMI}
        \begin{bmatrix}
            \bar{\alpha}^2P_\mathrm o & -\tilde{Z}_x^\top & Z_x^\top \\ -\tilde{Z}_x& 2\Lambda_\mathrm o-\tilde Z_s-\tilde Z_s^\top & {B_s^\mathrm o}^\top P_\mathrm o\\
            Z_x & P_\mathrm oB_s^\mathrm o & P_\mathrm o
        \end{bmatrix} \succ 0,
    \end{equation}
Therefore, contractivity and well-posedness of the untrained model~\eqref{eq:RNN_model} can be guaranteed by solving the LMI~\eqref{eq:contractivity_untrained_LMI} with decision variables $P_\mathrm o\in\mathbb S_+^{n}$, $\Lambda_\mathrm o\in\mathbb D_+^{\nu}$, $Z_x\in\R^{n\times n}$, $\tilde Z_x\in\R^{\nu\times n}$, and $\tilde Z_{s}\in\R^{\nu\times \nu}$, and then setting $A_x=P_\mathrm o^{-1}Z_x$, $\tilde A_x=\Lambda_\mathrm o^{-1}\tilde Z_x$, and $\tilde B _s^\mathrm o=\Lambda_\mathrm o^{-1}\tilde Z_s$. \\
Note that the contractivity property plays the same role as the echo state property in~\cite{jaeger2001echo}. In particular, this property is important during training, as it guarantees that the state trajectories of~\eqref{eq:RNN_model} asymptotically depend only on the driving input signals $(u,y)$, while the effect of the initial conditions vanishes asymptotically over time. Thanks to this property, we can leverage the linear-in-the-parameters structure of~\eqref{eq:RNN_model_y} with respect to the free parameters to formulate the learning problem as a convex optimisation problem.

%
A general drawback of the proposed approach is that fixing the matrices in \eqref{eq:RNN_model_x}-\eqref{eq:RNN_model_s} a priori reduces the number of free parameters, which may lead to lower performance compared with conventional gradient-based training methods. 
Note, however, that this limitation can be mitigated by resorting to hyperparameter selection procedures (see, e.g.,~\cite{armenio2020optimal} and~\cite{sgadari2026}) or based on prior knowledge of the plant dynamics. For instance, the contractivity bound $\bar{\alpha}$ in~\eqref{eq:contractivity_untrained_LMI} can be selected according to the expected time scales of the plant, thus restricting the admissible reservoir dynamics to compatible ones.
%
\subsection{The trained model}
In this work, model~\eqref{eq:RNN_model} is referred to as the untrained model. After the training of the free parameters $C$, $D$, and $D_s$ (see Section \ref{sec:training_alg} for details), one can use the output relation~\eqref{eq:RNN_model_y} into \eqref{eq:RNN_model_x}-\eqref{eq:RNN_model_s} and define $A=A_x+B_yC$, $B=B_u+B_yD$,  $B_s=B_s^\mathrm{o}+B_yD_s$, $\tilde A=\tilde A_x+\tilde B_yC$, $\tilde B=\tilde B_u+\tilde B_yD$, and $\tilde B_s=\tilde B_s^\mathrm{o}+\tilde B_yD_s$. In this way, we can rewrite~\eqref{eq:RNN_model} as
\begin{subequations}\label{eq:RNN_model_ss}
\begin{align}
    &x(k+1)=Ax(k)+Bu(k)+B_ ss(k)\\
    &s(k)=\sigma(\tilde Ax(k)+\tilde Bu(k)+\tilde B_ ss(k))\label{eq:RNN_model_ss_s}\\
    &y(k)=Cx(k)+Du(k)+D_ss(k)\label{eq:RNN_model_ss_y}
\end{align}
\end{subequations}
Note that, during the training, it is important to ensure that model \eqref{eq:RNN_model_ss} is well-posed, 
i.e., equation~\eqref{eq:RNN_model_ss_s} admits a unique solution $s(k)$ for any given pair $(x(k),u(k))$.  As discussed in~\cite{revay2023recurrent,ravasio2025developmentvelocityformclass}, a sufficient condition for model~\eqref{eq:RNN_model_ss} to be 
well-posed is the existence of a matrix $\Lambda\in\mathbb D_+^\nu$ such that \begin{equation}\label{eq:well_posedness_ss}
2\Lambda-\Lambda \tilde B_s-\tilde B_s^\top\Lambda\succ0.
\end{equation}
Notably, \eqref{eq:well_posedness_ss} will be imposed, through the inclusion of suitable matrix inequalities, in the training procedure.
%
%
\section{Learning unstructured models}\label{sec:training_alg}
As discussed, the approach proposed in this section is inspired by the one proposed in \cite{jaeger2001echo} for ESNs. In particular, considering model \eqref{eq:RNN_model}, the hyperparameters are previously defined as described in Section~\ref{sec:hyperparameters} and only matrices $C$, $D$, and $D_s$ are free training parameters. For notational reasons, we define $\theta = \begin{bmatrix}C&D&D_s\end{bmatrix}\in\R^{p\times r}$.\\
%
Since the reservoir hyperparameters are fixed, the reservoir dynamics \eqref{eq:RNN_model_x}-\eqref{eq:RNN_model_s} is completely specified. Therefore, 
to estimate the parameter vector $\theta$, we first simulate \eqref{eq:RNN_model_x}--\eqref{eq:RNN_model_s} from a random initial condition
$x(0)$ using the input-output sequence of the training dataset (see Assumption~\ref{ass:dataset}), i.e., setting
$u(k)=u_\mathrm d(k)$ and $y(k)=y_\mathrm d(k)$ for
$k=0,\ldots,N_\mathrm d$.
This yields the trajectories $\{x_\mathrm d(k)\}_{k=1}^{N_\mathrm d}$ and
$\{s_\mathrm d(k)\}_{k=1}^{N_\mathrm d}$, which represent the evolution of the variables $x(k)$ and $s(k)$, respectively.\\
In this section, we present two alternative algorithms for training unstructured models. Both approaches are formulated as convex optimisation problems leveraging the linear-in-the-parameters structure of \eqref{eq:RNN_model_ss_y}.
The first one, presented in Section~\ref{subsec:LS}, is based on reformulating the learning problem as a least-squares minimisation and provides a computationally efficient solution, suitable when computational complexity and fast model updates are the primary concerns. The second one, presented in Section~\ref{subsec:SM}, is based on set-membership identification and provides a rigorous method in the presence of bounded measurement noise. Moreover, this approach enables the computation of a feasible parameter set bounding the modelling uncertainty, which can be exploited in robust control schemes, along the lines of~\cite{d2025data}.\\
Therefore, whenever the optimisation problem is feasible for the selected reservoir hyperparameters, standard convex optimisation solvers are guaranteed to converge to a globally optimal solution.
\subsection{Least-squares learning procedure}
\label{subsec:LS}
In this section, we formulate the training as a least-squares optimisation problem, where the well-posedness condition~\eqref{eq:well_posedness_ss} is enforced through LMI constraints.\\
We define the regression matrix and the output matrix as
\begin{equation}
\begin{array}{cl}
\Phi \coloneq &
\begin{bmatrix}\label{eq:regressor_LS}
x_\mathrm d(\tau_\mathrm w+1)^\top &
u_\mathrm d(\tau_\mathrm w+1)^\top &
s_\mathrm d(\tau_\mathrm w+1)^\top\\
\vdots & \vdots & \vdots\\
x_\mathrm d(N_\mathrm d)^\top &
u_\mathrm d(N_\mathrm d)^\top &
s_\mathrm d(N_\mathrm d)^\top
\end{bmatrix}, \\[4mm]
Y_\mathrm d \coloneq &
\begin{bmatrix}
y_\mathrm d(\tau_\mathrm w+1)^\top\\
\vdots\\
y_\mathrm d(N_\mathrm d)^\top
\end{bmatrix},
\end{array}
\end{equation}
where the washout period $\tau_\mathrm w \in \mathbb{Z}_+$ is an additional hyperparameter introduced to accommodate the initial model transient due to the random initialisation.\\
At this point, the parameter vector $\theta$
can be estimated by solving the least-squares problem
\begin{equation}\label{opt:least_squares_identification}
\theta^\star=\arg\min_\theta J(\theta),
\end{equation}
where
\(J(\theta)
\coloneq
\frac{1}{N_\mathrm d-\tau_\mathrm w}
\left\| Y_\mathrm d - \Phi \theta^\top \right\|^2,\)
subject to the well-posedness condition \eqref{eq:well_posedness_ss}.
However, this results in a bilinear optimisation problem, which can be computationally intensive and prone to numerical issues.\\
To obtain a tractable formulation of this problem, we introduce the following proposition.
\begin{proposition}\label{prop:identification_lmi}
The optimisation problem 
\begin{subequations}\label{opt:identification_lmi}
\begin{align}
    &\min_{
        s\in \mathbb{R}_+,H\in\R^{p\times r}
    }
    s,\notag\\
    &\text{subject to: } \notag\\
    &\begin{bmatrix}\label{eq:diss_opt_cost_lmi}
        s &  H - (Q_\mathrm d^{-1}\Phi^\top Y_\mathrm d)^\top\tilde Q\\
        H^\top - \tilde Q Q_\mathrm d^{-1}\Phi^\top Y_\mathrm d & \tilde Q
    \end{bmatrix}\succeq 0
\end{align}
\end{subequations}
where $Q_\mathrm d=\Phi^\top\Phi$ and $\theta^\star=H\tilde Q^{-1}$, is equivalent to~\eqref{opt:least_squares_identification}, if there exists $\gamma\in\R_+$ such that 
\begin{equation}\label{eq:lmi_training_condition}
    \tilde Q=\gamma Q_\mathrm d.
\end{equation}
Moreover, define matrices $H_{\mathrm e}\in\R^{p\times n+m}$ and $H_s\in\R^{p\times \nu}$ such that $H=\left[H_{\mathrm e},\,H_s\right]$. If there exist matrices $\tilde  Q_{\mathrm e} \in \mathbb{S}_+^{n+m}$ and $Q_{s} \in \mathbb{D}_+^\nu$ such that $\tilde  Q=\diag(\tilde  Q_{\mathrm e},Q_{s})$, and the condition
\begin{equation}\label{eq:well_posedness_lmi}
2Q_{s}-\tilde B_s^\mathrm oQ_{s}-\tilde B_yH_s-Q_{s}\tilde{B}_s^{\mathrm{o}\,\top}-H_s^\top\tilde B_y^\top\succ0,
\end{equation}
holds, then model~\eqref{eq:RNN_model_ss} is well-posed.
\end{proposition}
The proof of Proposition \ref{prop:identification_lmi} is provided in the Appendix. This result provides an LMI reformulation of \eqref{opt:least_squares_identification} and a sufficient LMI condition for enforcing well-posedness of the learned model. 
However, the equivalence with the original least-squares problem holds only when condition~\eqref{eq:lmi_training_condition} is satisfied. In general, this condition cannot be enforced simultaneously with the well-posedness constraint~\eqref{eq:well_posedness_lmi}, and may therefore lead to an infeasible optimisation problem.\\ Inspired by~\cite{d2023virtual}, we relax~\eqref{eq:lmi_training_condition} by considering $\tilde Q_{\mathrm e}$ and $Q_{s}$ as optimisation variables and replacing~\eqref{eq:lmi_training_condition} with the following inequalities
\begin{equation}\label{eq:lmi_training_condition_relaxed}
\begin{aligned}
    &\tilde Q-\frac{1}{N_\mathrm d-\tau_\mathrm w}Q_\mathrm d+\lambda I_r\succeq0\\
    &-\tilde Q+\frac{1}{N_\mathrm d-\tau_\mathrm w}Q_\mathrm d+\lambda I_r\succeq0
\end{aligned}
\end{equation}
where $\lambda \in \mathbb{R}_+$ is a further optimisation variable to be minimised. Importantly, the well-posedness constraint~\eqref{eq:well_posedness_lmi} can still be explicitly enforced in the optimisation problem; therefore, this relaxation affects only the least-squares optimality condition and does not compromise the well-posedness guarantee of the learned model.\\
The resulting learning procedure is summarised in Algorithm~\ref{alg:LS_learning}, where $\beta\in\R_+$ is a design parameter. 
%
%
%
%
\begin{algorithm}
\small
\caption{Least\_squares\_learning}
\label{alg:LS_learning}
\begin{algorithmic}[1]
\Require $(\mathcal U_\mathrm d,\mathcal Y_\mathrm d)$, $\mathcal H$, $\tau_\mathrm w$, $\beta$
\Ensure $\theta^\star$
\State Simulate \eqref{eq:RNN_model_x}--\eqref{eq:RNN_model_s} from a random initial condition $x(0)$, using the available dataset input sequence $\mathcal U_\mathrm d$ and output sequence $\mathcal Y_\mathrm d$, and compute $(\Phi,Y_\mathrm d)$ as in \eqref{eq:regressor_LS}.  
\State Compute $\theta^\star=\text{WellPosed\_LS}(Y_\mathrm d,\Phi,r,\nu,p,B_s^\mathrm 0,\tau_\mathrm w,\beta)$
\end{algorithmic}
\end{algorithm}
\begin{algorithm}
\small
\caption{WellPosed\_LS}
\label{alg:LS_identification}
\begin{algorithmic}[1]
\Require $Y_\mathrm d$, $\Phi$, $r$, $\nu$, $p$, $B_s^\mathrm o$, $\tau_\mathrm w$, $\beta$
\Ensure $\theta^\star$
\State Solve
\begin{equation}\label{opt:identification_well_posedness_lmi}
\begin{aligned}
    &\min_{\substack{
        s,\lambda\in \mathbb{R}_+,\,\;
        \tilde Q_{\mathrm e} \in \mathbb{S}_+^{r-\nu},\,
        Q_{s} \in \mathbb{D}_+^\nu,\\
        H_{\mathrm e}\in\mathbb{R}^{p\times (r-\nu)},\,
        H_s\in\mathbb{R}^{p\times \nu}
    }}
    s+\beta \lambda\\
    &\text{subject to:}\\
    &\eqref{eq:diss_opt_cost_lmi},\;
      \eqref{eq:well_posedness_lmi},\;
      \eqref{eq:lmi_training_condition_relaxed}.
\end{aligned}
\end{equation}
\State Set
\(
\theta^\star = H \tilde Q^{-1}\), where
$H = [\,H_{\mathrm e},\, H_s\,]$ and
$\tilde Q = \operatorname{diag}(\tilde Q_{\mathrm e}, Q_{s})$.
\end{algorithmic}
\end{algorithm}
%
%
\subsection{The set-membership approach}
\label{subsec:SM}
The algorithm proposed in this section addresses the case in which biased and inaccurate parameter estimation may result when measurements are affected by additive bounded noise, i.e.
\begin{equation}\label{eq:output_rel_noise}
y_\mathrm d(k) = y(k) + \eta(k).
\end{equation}
where $\eta(k)$ verifies the following.
\begin{ass}\label{ass:noise_boundedness_SM}
    The noise satisfies $|\eta(k)|\leq\bar \eta$ for all $k\in\mathbb Z_+$, where $\bar \eta\in\R_+^{p}$ is known. 
\end{ass}
The main idea behind this approach, previously explored in~\cite{d2025data,sgadari2026}, is to structure the training into two steps. First, the set of all model parameterisations consistent with the available data, referred to as the feasible parameter set (FPS), is computed. Then, the parameterisation that achieves the best performance on a validation dataset according to a chosen suitability index is selected from this set. In the following, we extend this approach to the considered RNN model. Specifically, Section~\ref{sec:FPS_definition} characterises the FPS for the considered model, while Section~\ref{sec:FPS_scenario_sampling} introduces an LMI-based approach that enables the integration of the well-posedness condition~\eqref{eq:well_posedness_lmi} into the FPS sampling procedure.\\ 
In this case we assume also that a validation dataset is available, which satisfies the following.
\begin{ass}\label{ass:dataset_val}
    A validation dataset independent from the one used for training is available. This dataset consists of an input sequence $\mathcal U_\mathrm v=\{u_\mathrm v(k)\}_{k=1}^{N_\mathrm v}$ and a measured output sequence $\mathcal Y_\mathrm v=\{y_\mathrm v(k)\}_{k=1}^{N_\mathrm v}$.
\end{ass}
\subsubsection{Definition of the FPS}\label{sec:FPS_definition}
To define the FPS, we introduce the vectors $\phi_\mathrm d(k) \coloneq [x_\mathrm d(k)^\top,\; u_\mathrm d(k)^\top,\;s_\mathrm d(k)]^\top$ and $\phi(k) \coloneq [x(k)^\top,\; u(k)^\top,\;s(k)^\top]^\top$, where $x_\mathrm d(k)$ and $s_\mathrm d(k)$ (respectively, $x(k)$ and $s(k)$) are obtained by simulating model~\eqref{eq:RNN_model} with the dataset sequences $(\mathcal{U}_\mathrm d,\mathcal{Y}_\mathrm d)$ (respectively, the non-noisy sequences $\{u(k),y(k)\}_{k=0}^{N_{\mathrm d}}$). Since $u(k)=u_\mathrm d(k)$ for all $k$, \eqref{eq:output_rel_noise} can be rewritten as
\[y_\mathrm d^{(i)}(k)=\theta^{(i)}\phi_\mathrm d(k)+\epsilon_i(k)+\eta^{(i)}(k),\]
for all $i=1,\dots, p$, where 
\begin{equation}\label{eq:SM_state_bound}
\epsilon_i(k) \coloneq -\theta^{(i)} w_\mathrm d(k),
\end{equation} accounts for the effect of the measurement noise on the state predictions, and
\[w_\mathrm d(k)\coloneq \phi_\mathrm d(k)-\phi(k)=\begin{bmatrix}
    x_\mathrm d(k)-x(k)\\0\\s_\mathrm d(k)-s(k)
\end{bmatrix}\in\R^r.\]\\
The following proposition can be proved.
\begin{proposition}\label{prop:contarctivity_implies_diss}
Under Assumptions~\ref{ass:sigmoid_function} and \ref{ass:noise_boundedness_SM}, if the RNN~\eqref{eq:RNN_model} satisfies condition~\eqref{eq:RENcontractivity_Lyapunov}, then:
\begin{itemize}
    \item[(i)] model~\eqref{eq:RNN_model} is $\delta$ISS;
    \item[(ii)] there exist functions $\beta\in\mathcal{KL}$ and $\gamma\in\mathcal{K}_\infty$ such that for any $k\in\mathbb Z_+$, it holds that
    \[\norm{x_\mathrm d(k)-x(k)}\leq\bar w_x(k)\coloneq\beta(\norm{x_\mathrm d(0)-x(0)},k)+\gamma(\bar \eta);\]
    \item[(iii)] for any $k\in\mathbb Z_+$, there exist $\Sigma(k)\in\mathbb D_+^{\nu}$ where $\Sigma(k)\preceq I_\nu$ such that $\norm{s_\mathrm d(k)-s(k)}\leq \bar w_s(k)$, where \begin{multline*}\bar w_s(k)\coloneq\norm{(I_\nu-\Sigma(k)\tilde B_s^\mathrm o)^{-1}\Sigma(k)\tilde A_x}\bar w_x(k)\\+\norm{(I_\nu-\Sigma(k)\tilde B_s^\mathrm o)^{-1}\Sigma(k)\tilde B_y}\bar\eta.\end{multline*}
\end{itemize}
\end{proposition}
The proof of Proposition~\ref{prop:contarctivity_implies_diss} has been postponed to the appendix for clarity reasons. 
In the light of it and \eqref{eq:SM_state_bound}, if the model hyperparameters are chosen as described in Section~\ref{sec:hyperparameters}, then  $|\epsilon_i(k)|\leq\bar\epsilon_i\coloneq \norm{C^{(i)}}\bar w_x(k)+ \norm{D_s^{(i)}}\bar w_s(k)$ for all $k\in\mathbb Z_+$ and for all $i=1,\dots,p$.\\
Exploiting the boundedness of $\epsilon_i$, the FPS $\Theta$ can be defined following the procedure outlined in Algorithm~~\ref{alg:FPS_computation} \cite{d2025data}.\\
\begin{algorithm}[h!]
\small
\caption{FPS\_computation}
\label{alg:FPS_computation}
\begin{algorithmic}[1]
\Require $\{y_\mathrm d(k)\}_{k=\tau_\mathrm w+1}^{N_\mathrm d}$, $\{\phi_\mathrm d(k)\}_{k=\tau_\mathrm w+1}^{N_\mathrm d}$, $r$, $p$, $\tau_\mathrm w$, $\bar{\eta}$
\Ensure $\Theta$
\For{$i = 1,\dots,p$}
    \State Solve the optimisation problem
    \begin{align*}
        \underline{\lambda}_i
        =\;&\min_{\lambda \in \mathbb{R}_{\ge 0},\, K \in \mathbb{R}^{r}} \;\lambda\\[2mm]
        &\text{subject to}\\
        &\bigl|y^{(i)}_{\mathrm d}(k) - K^\top \phi_\mathrm d(k)\bigr|
        \;\le\; \lambda + \bar{\eta}^{(i)},\\
        &\qquad \forall\, k \in \{\tau_\mathrm w+1,\dots,N_\mathrm d\}.
    \end{align*}
    \State Compute
    \begin{equation}\label{eq:SM_bound_estimate}
        \hat{\bar{\epsilon}}_i = \alpha_i \underline{\lambda}_i,
        \quad \alpha_i > 1.
    \end{equation}
\EndFor
\State Construct the FPS
\(
\Theta = \Theta_1(\alpha_1) \times \dots \times \Theta_p(\alpha_p),
\)
where, for all $k=\tau_\mathrm w+1,\dots,N_\mathrm d$,
\[
\Theta_i(\alpha_i)
=
\left\{
\theta^{(i)} \in \mathbb{R}^{1\times r}
\,:\,
\bigl|y_\mathrm d^{(i)}(k)
-
\theta^{(i)} \phi_\mathrm d(k)\bigr|
\le
\hat{\bar{\epsilon}}_i + \bar{\eta}^{(i)}
\right\}.
\]
\end{algorithmic}
\end{algorithm}
The parameter $\alpha_i$ in \eqref{eq:SM_bound_estimate} accounts for the uncertainty arising from the finite number of measurements in the dataset, and satisfies $\alpha_i \to 1^+$ as $N_\mathrm{d}$ increases. 
A practical way to define $\alpha_i$ is to set this parameter, for all $i=1,\dots,p$, at the minimum value such that the least square model parametrisation $\theta_{\mathrm{LS}}$, lies within $\Theta$. In particular, we set, for all $i=1,\dots,p$ \cite{sgadari2026},
\begin{equation}\label{eq:SM_alpha_i}
    \alpha_i=\max_{k\in\{\tau_\mathrm w+1,\dots,N_\mathrm d\}}\cfrac{|y^{(i)}(k)-{\theta_{\mathrm{LS}}^{(i)}}\phi_\mathrm d(k)|-\bar\eta^{(i)}}{\underline{\lambda}_i}.
\end{equation}
If, however, \eqref{eq:SM_alpha_i} yields an FPS such that $\theta_\mathrm{LS}\notin\Theta$, this outcome indicates a possibly incorrect choice of the model class, and the procedure should therefore be repeated with different hyperparameters.
\subsubsection{Scenario sampling of the FPS}~\label{sec:FPS_scenario_sampling}
Now that the FPS is defined, we need to determine the optimal model parameterisation $\theta^\star \in \Theta$ that achieves the highest performance on the validation dataset.\\
Since exploring all possible values $\theta \in\Theta$ may be computationally intractable, we restrict the analysis to $N_\mathrm s$ scenarios $\theta^{[t]}$, for $t=1,\dots,N_\mathrm s$, drawn from $\Theta$. 
For each scenario, model~\eqref{eq:RNN_model_ss} is simulated using the validation dataset input sequence $\mathcal U_\mathrm v$, resulting in the output simulated trajectory $\{y^{[t]}(k)\}_{k=1}^{N_\mathrm v}$. To assess the model performance of each scenario, following the approach in~\cite{d2025data}, we compute the minimum distance of the simulated output of scenario $t$ from the noisy output data tube $\tilde{\mathcal Y}(k)$ as
\begin{equation}\label{eq:set_distance_index}
d^{[t]}\coloneq\min_{k=\tau_\mathrm w+1,\dots,N_\mathrm v}\text{dist}({y^{[t]}(k),\tilde{\mathcal Y}(k)}),
\end{equation}
where
\[\tilde{\mathcal Y}(k)\coloneq\{\tilde y(k)\in\R^p\,:\,|\tilde y^{(i)}(k)-y_\mathrm v^{(i)}(k)|\leq\bar\eta^{(i)},\,\forall i=1,\dots,p\}.\]
The optimal parameter value $\theta^\star$ is then selected as the one associated with the minimum distance $d^{\star}\coloneq\min_{t=1,\dots,N_\mathrm s}d^{[t]}$.\\
The following result provides a criterion for selecting the number of scenarios $N_\mathrm s$.
\begin{proposition}[ \cite{d2025data}]
    Let $\theta^{[N_\mathrm s+1]}\in \Theta$ be a random matrix with probability distribution $\mathbb P_{\theta^{[N_\mathrm s+1]}}$ over $\Theta$. Also, let $\epsilon\in(0,1)$ and $\beta\in(0,1)$ be two user-defined constants. For all $N_\mathrm s\geq1$ such that $N_\mathrm s\geq\log_{1-\epsilon}(\beta)$, then with probability $1-\beta$ it holds that $\mathbb P\{\theta^{[N_\mathrm s+1]}\in\Theta\,:\,d^{[N_\mathrm s+1]}< d^{\star}\}\leq\epsilon$.
\end{proposition}
As a practical consideration, sampling $\theta$ from the FPS can be performed independently for each row, i.e., by sampling $\theta^{(i) }\in(\Theta_i (\alpha_i ),\mathbb P_{\theta,i})$ for $i=1,\dots,p$. Each distribution $\mathbb P_{\theta,i}$ can be defined with a mean equal to the least-squares estimator $\theta_\mathrm{LS}^{(i)}$ and a covariance equal to the corresponding covariance matrix.

However, this approach does not generally ensure that the model extracted from the FPS satisfies the well-posedness condition \eqref{eq:well_posedness_ss}.
Given $\theta^{[t]}\in \Theta$, the closest feasible parameter $\tilde\theta^{[t]}$ satisfying~\eqref{eq:well_posedness_ss} can be computed by solving the LMI problem
\begin{subequations}\label{opt:SM_scenario}
{\small
\begin{align}
    &\min_{\substack{\{c_i\in\R_+\}_{i=1}^p,\,\;
        \tilde Q_{\mathrm e} \in \mathbb{S}_+^{r-\nu},\,
        Q_{s} \in \mathbb{D}_+^\nu,\\
        H_{\mathrm e}\in\mathbb{R}^{p\times (r-\nu)},\,
        H_s\in\mathbb{R}^{p\times \nu}}}\sum_{i=1}^{p}c_i\\
    &\text{subject to:}\notag\\
    &\begin{bmatrix}
        c_i & {\theta^{[t]}}^{(i)}\tilde Q - H^{(i)}\\
        ({\theta^{[t]}}^{(i)}\tilde Q - H^{(i)})^\top &\tilde Q
    \end{bmatrix}\succeq0,\,\forall i=1,\dots,p \label{eq:SM_cstr}\\
    &\text{LMI }\,\eqref{eq:well_posedness_lmi}\notag
\end{align}}
\end{subequations}
and setting $\tilde\theta^{[t]}=H\tilde Q^{-1}$, where $H=[\,H_\mathrm e,\,H_s]$ and $\tilde Q=\diag(\tilde Q_\mathrm e,Q_{s})$.
Note that minimising $c_i$ under constraint~\eqref{eq:SM_cstr} is equivalent to minimising the weighted norm $ \norm{{\theta^{[t]}}^{(i)} - {\tilde\theta^{[t]^{(i)}}}}_{\tilde Q}
$. 
The overall set-membership learning procedure is summarised in Algorithm~\ref{alg:set_membership_learning}.
\begin{algorithm}
\small
\caption{Set\_membership\_learning}
\label{alg:set_membership_learning}
\begin{algorithmic}[1]
\Require $(\mathcal U_\mathrm d,\mathcal Y_\mathrm d)$, $(\mathcal U_\mathrm v,\mathcal Y_\mathrm v)$, $\mathcal H$, $\tau_\mathrm w$, $\bar \eta$
\Ensure $\theta^\star$
\State Simulate \eqref{eq:RNN_model_x}--\eqref{eq:RNN_model_s} from a random initial condition $x(0)$, using the dataset input sequence $\mathcal U_\mathrm d$ and output sequence $\mathcal Y_\mathrm d$.  
This yields the trajectories $\{x_\mathrm d(k)\}_{k=1}^{N_\mathrm d}$ and $\{s_\mathrm d(k)\}_{k=1}^{N_\mathrm d}$.


\State Compute the FPS 
\[
\begin{aligned}\Theta =& \text{FPS\_computation}
(\{y_\mathrm d(k)\}_{k=\tau_\mathrm w+1}^{N_\mathrm d}, \{\phi_\mathrm d(k)\}_{k=\tau_\mathrm w+1}^{N_\mathrm d},\\& r, p, \tau_\mathrm w, \bar{\eta})\end{aligned}.
\]
\ForAll{$t=1,\dots,N_\mathrm s$} 
    \State Extract a parameter sample $\theta^{[t]}$ ensuring that the model is well-posed
    \[\theta^{[t]}=\text{WellPosed\_scenario\_sampling}(\Theta,t,\nu,p,\tilde B_s^\mathrm o)\]
    
    \If{$\tilde{\theta}^{[t]} \notin \Theta$}
        \State Set $d^{[t]} = M$, where $M$ is a large positive scalar.
    \Else
        \State Simulate model \eqref{eq:RNN_model_ss} using the dataset input sequence $\mathcal U_\mathrm v$, resulting in the output trajectory
        $\{y^{[t]}(k)\}_{k=1}^{N_\mathrm v}$.
        \State Compute $d^{[t]}$ according to \eqref{eq:set_distance_index}.
    \EndIf
\EndFor

\State Select $\theta^\star = \theta^{[t^\star]}$, where
\(t^\star \coloneq \arg\min_{t \in \{1,\dots,N_\mathrm s\}} d^{[t]}.
\)
\end{algorithmic}
\end{algorithm}
\begin{algorithm}
\small
\caption{WellPosed\_scenario\_sampling}
\label{alg:scenario_parameter_selection}
\begin{algorithmic}[1]
\Require $\Theta$, $r$, $\nu$, $p$, $\tilde B_{s}^\mathrm o$
\Ensure $\tilde \theta^\mathrm s$

\State Randomly extract $\theta^{\mathrm s}$ from $(\Theta,\mathbb{P}_{\theta^{\mathrm s}})$.

\State Solve~\eqref{opt:SM_scenario}.

\State Set
    \(
    \tilde{\theta}^\mathrm s = H \tilde Q^{-1},
    \)
    where $H=[H_\mathrm e,H_s]$ and $\tilde Q=\diag(\tilde Q_\mathrm e,Q_{s})$.
\end{algorithmic}
\end{algorithm}
\section{Learning structured models}
\label{sec:learning structured models}
In this section we address the design of a physics-informed procedure to derive a modular plant model, where the modularity is inspired by that of the plant $\mathcal{P}$. In particular, we discuss how to identify a number $n_\mathrm s$ of submodels $\mathcal{M}_i$, each corresponding with a subplant $\mathcal{P}_i$, with $i=1,\dots,n_\mathrm s$, each having as local input and output the pair $(u_i,y_i)$ and where interconnections with the other submodels occur through suitable interconnection variables $\nu^{\mathcal{M}}_{ij}$, for $j\neq i$. To this regard, the main modelling choice lies in the twofold selection of (i) the model interconnection network, and (ii) the coupling variables.\\
Regarding (i), we need to define which submodels have a direct influence on $\mathcal{M}_i$ for all $i\in\mathcal{I}$, i.e., the values of $j\neq i$ such that $\nu_{ij}^{\mathcal{M}}\neq 0$. At the same time, problem (ii) requires to define how $\nu_{ij}^{\mathcal{M}}$, when not identically equal to zero, is composed. Formally speaking, the scope is to define, for all $i\in\mathcal I$, the neighboring sets: $\mathcal{N}^{\mathcal{M}}_{u,i}\coloneq\{j\in\mathcal I\setminus\{i\}:\nu^{\mathcal{M}}_{ij}\neq 0\text{ includes entries of }u_j\}$ and
$\mathcal{N}^{\mathcal{M}}_{x,i}\coloneq\{j\in\mathcal I\setminus\{i\}:\nu^{\mathcal{M}}_{ij}\neq 0\text{ includes state variables of }\mathcal{M}_j\}$.\\ 
As also discussed in Section \ref{subsec:modular_system}, in this work we learn models with a sparse interconnection structure that reflects the topology of the plant $\mathcal{P}$ by retaining only the direct physical links between subplants. 
\subsection{The proposed modular learning approach}
The proposed approach requires to set
$\mathcal{N}^{\mathcal{M}}_{u,i}=\mathcal{N}^{\mathcal{P}}_{u,i}$ and 
$\mathcal{N}^{\mathcal{M}}_{x,i}=\mathcal{N}^{\mathcal{P}}_{x,i}$ and to compose also the state $x$ and vector $s$ of~\eqref{eq:RNN_model} by $n_\mathrm{s}$ non-overlapping sub-vectors $x_i \in \mathbb{R}^{n_i}$ and $s_i \in \mathbb{R}^{\nu_i}$, respectively, where $\sum_{i=1}^{n_\mathrm{s}} n_i = n$ and $\sum_{i=1}^{n_\mathrm{s}} \nu_i = \nu$.\\
As discussed, this modelling choice is consistent with non-overlapping decompositions adopted in a model-based framework~\cite{farina2018distributed}. On the one hand, the dimensionality of the submodels is reduced, leading to improved interpretability and lower computational cost. On the other hand, as a downside, enforcing on the RNN a structure consistent with the interconnection pattern of the underlying plant inevitably introduces approximations. In particular, it is important to remark that we may perform an approximation every time we set $\nu_{ij}^{\mathcal M}=0$ if and only if $\nu_{ij}^{\mathcal P}=0$ and that such approximation
highly depends upon the sampling time \footnote{To understand this statement, recall that $\nu_{ij}^{\mathcal P}=0$ means that there is no direct connection between the variables of the subplants $\mathcal{P}_j$ and $\mathcal{P}_i$, i.e., that, in a physics-based continuous-time mathematical model of the plant, the internal variables of plant $\mathcal{P}_j$ do not appear in the dynamics of $\mathcal{P}_i$. However, a direct path between $\mathcal{P}_j$ and $\mathcal{P}_i$ may be present in $\mathcal{G}^{\mathcal{P}}$: in this case, in the corresponding discrete-time model obtained by zero-order-hold discretization, the internal variables of plant $\mathcal{P}_j$ will appear in the dynamics of $\mathcal{P}_i$. This is discussed, for the linear case, in~\cite{farina2013block}. As discussed in~\cite{farina2013block}, the derivation of modular discrete-time models suitable for decentralised and distributed controller design requires the adoption of approximation methods (e.g., the so-called mixed Euler-ZOH), which introduce approximation errors that vanish only if the sampling time $T_\mathrm s\to0$. This consideration sheds some light also on learning (sampled data-based) structured models, since it clarifies that any structured discrete-time model is an approximation, where the approximation error can be reduced by reducing the sampling time, whose choice becomes critical.}.\smallskip\\
To embed the desired interconnection structure into the model \eqref{eq:RNN_model}, first we select the matrices in \eqref{eq:RNN_model_x}--\eqref{eq:RNN_model_s} with a block-diagonal structure, i.e.,
\begin{equation}\label{eq:hyperparam_structured}
\begin{aligned}
&A= \diag(A_1,\dots, A_{n_\mathrm{s}}), 
\!\!&\!\! &B= \diag(B_1,\dots, B_{n_\mathrm{s}}),\\
&B_s= \diag(\tilde B_{s,1},\dots, \tilde B_{s,n_\mathrm{s}}), 
\!\!&\!\! &B_y= \diag(\tilde B_{y,1},\dots, \tilde B_{y,n_\mathrm{s}}),\\
&\tilde A= \diag(\tilde A_1,\dots, \tilde A_{n_\mathrm{s}}), 
\!\!&\!\! &\tilde B = \diag(\tilde B_1,\dots, \tilde B_{n_\mathrm{s}}),\\
&\tilde B_s = \diag(\tilde B_{s,1},\dots, \tilde B_{s,n_\mathrm{s}}), 
\!\!&\!\! &\tilde B_y = \diag(\tilde B_{y,1},\dots, \tilde B_{y,n_\mathrm{s}}).
\end{aligned}
\end{equation}
where $A_i \in \mathbb{R}^{n_i \times n_i}$, 
$B_i \in \mathbb{R}^{n_i \times m_i}$, 
$B_{s,i} \in \mathbb{R}^{n_i \times \nu_i}$, 
$B_{y,i} \in \mathbb{R}^{n_i \times p_i}$, 
$\tilde A_i \in \mathbb{R}^{\nu_i \times n_i}$, 
$\tilde B_i \in \mathbb{R}^{\nu_i \times m_i}$, 
$\tilde B_{s,i} \in \mathbb{R}^{\nu_i \times \nu_i}$, and 
$\tilde B_{y,i} \in \mathbb{R}^{\nu_i \times p_i}$, 
for all $i \in \mathcal I$.
On the other hand, the free parameters are structured in accordance with the interconnection graph, i.e.,
\begin{equation}\label{eq:free_param_structured}
C \in \mathcal C, \quad D \in \mathcal D, \quad D_s \in \mathcal D_s,
\end{equation}
where
\begin{align*}
\mathcal C \coloneqq 
\Big\{
C \in \mathbb{R}^{p \times n} \,:\,
C = [C_{i,j}]_{i,j \in \mathcal I},\,
C_{i,j} \in \mathbb{R}^{p_i \times n_j},\\
C_{i,j} = 0 \ \text{if } j \notin \mathcal N_{x,i}^{\mathcal M}\cup\{i\}
\Big\}, \\
\mathcal D \coloneqq 
\Big\{
D \in \mathbb{R}^{p \times m} \,:\,
D = [D_{i,j}]_{i,j \in \mathcal I},\,
D_{i,j} \in \mathbb{R}^{p_i \times m_j},\\
D_{i,j} = 0 \ \text{if } j \notin \mathcal N_{u,i}^{\mathcal M}\cup\{i\}
\Big\}, \\
\mathcal D_s \coloneqq 
\Big\{ D_s \in \mathbb R^{p\times \nu} \ : \ 
D_s = \diag(D_{s,1},\dots,D_{s,n_\mathrm s}), \\ 
D_{s,i}\in\mathbb R^{p_i\times \nu_i}
\Big\}.
\end{align*}
At a submodel level, this results in 
\begin{subequations}\label{eq:PI_RNN_subsystem_i}
{\small
    \begin{align}
        &x_i(k+1)\!=\!A_{x,i}x_i(k){+}B_{u,i}u_i(k){+}B_{s,i}^\mathrm os_i(k)+B_{y,i}y_i(k)\label{eq:PI_RNN_subsystem_i_x}\\
        &s_i(k)=\sigma_i(\tilde A_{x,i}x_i(k){+}\tilde B_{u,i}u_i(k){+}\tilde B_{s,i}^\mathrm os_i(k){+}\tilde B_{y,i}y_i(k))\label{eq:PI_RNN_subsystem_i_s}\\
        &y_i(k)= \!\!\!\! \sum_{j\in\mathcal{N}^{\mathcal{M}}_{x,i}\cup\{i\}}\!\!\!\!C_{i,j}x_j(k)+\!\!\!\!\sum_{j\in\mathcal{N}^{\mathcal{M}}_{u,i}\cup\{i\}}\!\!\!\!D_{i,j}u_j(k)+D_{s,i}s_i(k)
        \label{eq:PI_RNN_subsystem_i_y}
    \end{align}}
\end{subequations}
for each $i\in\mathcal{I}$, where $\sigma_i(\cdot):\R^{\nu_i}\to\R^{\nu_i}$. Specifically, the state dynamics of the submodels are fully decoupled, while the output of each submodel depends on its own state and input and on the state and input of its neighbors. 
\subsection{The learning algorithm}
In this context, we replace Assumptions~\ref{ass:dataset} and \ref{ass:dataset_val} with the following one.
\begin{ass}\label{ass:dataset_local}
Each subplant $\mathcal{P}_i$, for $i \in \mathcal{I}$, has access to a training dataset consisting of an applied input sequence
$\mathcal U_{\mathrm d,i} = \{u_{\mathrm d,i}(k)\}_{k=1}^{N_\mathrm d}$
and a measured output sequence
$\mathcal Y_{\mathrm d,i} = \{y_{\mathrm d,i}(k)\}_{k=1}^{N_\mathrm d}$,
and to a validation dataset consisting of an applied input sequence
$\mathcal U_{\mathrm v,i} = \{u_{\mathrm v,i}(k)\}_{k=1}^{N_\mathrm v}$
and a measured output sequence
$\mathcal Y_{\mathrm v,i} = \{y_{\mathrm v,i}(k)\}_{k=1}^{N_\mathrm v}$.
\end{ass}
From the structural perspective,~\eqref{eq:PI_RNN_subsystem_i} is similar to~\eqref{eq:RNN_model}: each subplant $\mathcal{P}_i$ has to perform the procedure described in Algorithm~\ref{alg:distributed_identification}, which essentially corresponds with the procedure described in Section~\ref{sec:model}, parallelised across the agents.
\begin{algorithm}
\small
\caption{Distributed\_ structured\_learning}
\label{alg:distributed_identification}
\begin{algorithmic}[1]
\Require $\mathcal I$, $\{\mathcal N^{\mathcal M}_{u,i},\mathcal N^{\mathcal M}_{x,i},(\mathcal U_{\mathrm d,i},\mathcal Y_{\mathrm d,i}),(\mathcal U_{\mathrm v,i},\mathcal Y_{\mathrm v,i})\}_{i\in\mathcal I}$, $\mathcal H$, $\tau_{\mathrm w}$, $\beta$, $\bar{\eta}$
\Ensure $\{\theta_i^\star\}_{i\in\mathcal I}$
\ForAll{$i \in \mathcal I$}


\State Simulate the local state dynamics
\eqref{eq:PI_RNN_subsystem_i_x}--\eqref{eq:PI_RNN_subsystem_i_s}
from a random initial condition $x_i(0)$, using the local dataset
$(\mathcal U_{\mathrm d,i},\mathcal Y_{\mathrm d,i})$, resulting in the trajectories
$\{x_{\mathrm d,i}(k)\}_{k=0}^{N_\mathrm d}$ and
$\{s_{\mathrm d,i}(k)\}_{k=0}^{N_\mathrm d}$.

\State Define the local data vectors
\[
\begin{array}{cccc}
X_{\mathrm{d},i}\coloneq
\begin{bmatrix}
x_{\mathrm{d},i}(\tau_\mathrm w+1)^\top\\
\vdots\\
x_{\mathrm{d},i}(N_\mathrm d)^\top
\end{bmatrix},
&
U_{\mathrm{d},i}\coloneq
\begin{bmatrix}
u_{\mathrm{d},i}(\tau_\mathrm w+1)^\top\\
\vdots\\
u_{\mathrm{d},i}(N_\mathrm d)^\top
\end{bmatrix},
\\[4mm]
S_{\mathrm{d},i}\coloneq
\begin{bmatrix}
s_{\mathrm{d},i}(\tau_\mathrm w+1)^\top\\
\vdots\\
s_{\mathrm{d},i}(N_\mathrm d)^\top
\end{bmatrix},
&
Y_{\mathrm{d},i}\coloneq
\begin{bmatrix}
y_{\mathrm{d},i}(\tau_\mathrm w+1)^\top\\
\vdots\\
y_{\mathrm{d},i}(N_\mathrm d)^\top
\end{bmatrix}.
\end{array}
\]

\State Receive the sequences $U_{\mathrm{d},j}$ and $X_{\mathrm{d},j}$ from all
$j \in \mathcal N^{\mathcal M}_{u,i}$ and
$j \in \mathcal N^{\mathcal M}_{x,i}$, respectively.

\State Define the local free parameter
\[
\theta_i \coloneq
\begin{bmatrix}
(C)_{\{i\},\mathcal N^{\mathcal M}_{x,i}\cup\{i\}}&
(D)_{\{i\},\mathcal N^{\mathcal M}_{u,i}\cup\{i\}}&
D_{s,i}
\end{bmatrix}
\in \mathbb{R}^{p_i \times r_i}.
\]

\If{the least-squares approach is selected}
    \State Define
    \[
    \Phi_i \coloneq
    \begin{bmatrix}
    \left[X_{\mathrm{d},j}\right]_{j \in \mathcal N^{\mathcal M}_{x,i}\cup\{i\}}&
    \left[U_{\mathrm{d},j}\right]_{j \in \mathcal N^{\mathcal M}_{u,i}\cup\{i\}}&
    S_{\mathrm{d},i}
    \end{bmatrix}.
    \]
    \State Compute
    \(
    \theta_i^\star=\) 
    WellPosed\_LS
    \((
    Y_{\mathrm d,i},
    \Phi_i,     
    r_i,
    \nu_i,p_i,\) \(
    B_{s,i}^{\mathrm o},
    \tau_\mathrm w,
    \beta
    ).
    \)
\Else
    \State Define, for all $k=\tau_\mathrm w+1,\dots,N_\mathrm d$, \(\phi_{\mathrm d,i}(k) \coloneq \begin{bmatrix}x_{\mathrm d,i}(k)^\top& u_{\mathrm d,i}(k)^\top&s_{\mathrm d,i}(k)^\top\end{bmatrix}\)
    \State Compute the local FPS
    \[
\begin{aligned}
\Theta_i &=
\text{FPS\_computation}\bigl(
\{y_{\mathrm d,i}(k)\}_{k=\tau_\mathrm w+1}^{N_\mathrm d},
\{\phi_{\mathrm d,i}(k)\}_{k=\tau_\mathrm w+1}^{N_\mathrm d}, \\
&\qquad
r_i,
p_i,
\tau_{\mathrm w},
\bar{\eta}_i
\bigr),
\end{aligned}
    \]
    where $\bar{\eta}_i\in\R^{p_i}$ is the noise bound associated with $y_i$.
\EndIf

\EndFor

\If{the set-membership approach is selected}
    \ForAll{$t\in\{1,\dots,N_\mathrm s\}$}
        \ForAll{$i\in\mathcal I$}
            \State Extract 
            a sample $\theta_i^{[t]}$ ensuring the well-posedness of $\mathcal M_i$
            \[\theta_i^{[t]}=\text{WellPosed\_scenario\_sampling}(\Theta_i,r_i,\nu_i,p_i,\tilde B_{s,i}^\mathrm o)\]
        \EndFor
         \If{$\exists\, i \in \mathcal I$ such that $\theta_i^{[t]} \notin \Theta_i$}
            \State Set $d^{[t]} = M$, where $M$ is a large positive scalar.
        \Else
            \State Simulate in parallel the submodels $\mathcal M_i$, for all $i \in \mathcal I$,
            using the local input sequences $\{\mathcal U_{\mathrm v,i}\}_{i=1}^{n_s}$,
            resulting in the output trajectory $\{y^{[t]}(k)\}_{k=0}^{N_\mathrm v}$.
            \State Compute $d^{[t]}$ according to~\eqref{eq:set_distance_index}.
        \EndIf
        
    \EndFor
    \State Select $\theta^\star_i=\theta_i^{[t^\star]}$ for all $i\in\mathcal I$.
\EndIf

\end{algorithmic}
\end{algorithm}

After identification of the free parameter vector $\theta_i$, we define matrices $A_i=A_{x,i}+B_{y,i}C_{i,i}$, $B_i=B_{u,i}+B_{y,i}D_{i,i}$, $B_{s,i}=B_{s,i}^o+B_{y,i}D_{s,i}$, $\tilde A_i=\tilde A_{x,i}+\tilde B_{y,i}C_{i,i}$, $\tilde B_i=\tilde B_{u,i}+\tilde B_{y,i}D_{i,i}$, $\tilde B_{s,i}=\tilde B_{s,i}^o+\tilde B_{y,i}D_{s,i}$, $B_{w,i}=B_{y,i}$, and $ \tilde B_{w,i}=\tilde B_{y,i}$, for all $i\in\mathcal{I}$. The so-obtained submodel $\mathcal{M}_i$ is therefore 
\begin{equation}\label{eq:non_overlapping_ss}
\begin{aligned}
&x_i(k+1) = A_{i}x_i(k) {+} B_{i}u_i(k) {+} B_{s,i}s_i(k) {+} B_{w,i}w_i(k)\\
&s_i(k) = \sigma_i\left(\tilde A_{i}x_i(k) {+} \tilde B_{i}u_i(k) {+} \tilde B_{s,i}s_i(k) {+} \tilde B_{w,i}w_i(k) \right)\\
&y_i(k) = C_{i,i}x_i(k)+D_{i,i}u_i(k)+D_{s,i}s_i(k) + w_i(k)
\end{aligned}
\end{equation}
where the term
\[w_i(k)=\sum_{j\in\mathcal{N}^{\mathcal{M}}_{y,i}} C_{i,j}x_{j}(k)+\sum_{j\in\mathcal{N}^{\mathcal{M}}_{u,i}}D_{i,j}u_{j}(k),\]
accounts for the effect of physical couplings of neighbouring subplants.\smallskip\\
The learning procedure described in Algorithm~\ref{alg:distributed_identification} is inherently scalable. In fact, the computational complexity grows with the size of the individual submodel rather than with the size of the full model. Additionally, data exchange occurs only between neighbouring subplants and involves only input and state information, making the approach well-suited to settings with communication and privacy constraints. Note that the RNN state generally does not correspond to physically meaningful (and therefore sensitive) quantities. In contrast, output data, which are usually more sensitive, are not shared. Furthermore, since $n_i\leq n$, for all $i\in\mathcal {I} $, this approach yields reduced-order models, thereby significantly reducing the number of decision variables when the model is employed in a decentralised or distributed control scheme. 
%
%
%
%
\section{Learning models with stability guarantees}\label{sec:diss_learning}
In this section we address the design of physics-informed procedures for deriving a plant model that enjoys the $\delta$ISS property.
\subsection{Learning unstructured stable models}\label{sec:diss_identification}
In this section we discuss how the procedures presented in Section~\ref{sec:training_alg} for unstructured models can be modified to impart $\delta$ISS to~\eqref{eq:RNN_model_ss}. To do this, the following proposition is required.
\begin{proposition}\label{prop:diss_opt}
Consider the untrained model~\eqref{eq:RNN_model} and let Assumption~\ref{ass:sigmoid_function} hold. If there exist matrices $H_x\in\R^{p\times n}$, $H_u\in\R^{p\times m}$, $H_s\in\R^{p\times \nu}$, $Q_\mathrm C,\tilde Q_x \in \mathbb{S}_+^n$, $Q_\mathrm D, \tilde Q_u \in \mathbb{S}_+^m$,and $Q_{s}\in\mathbb{D}_+^\nu$, such that the condition
\begin{equation}\label{eq:diss_lmi}
\resizebox{\columnwidth}{!}{$
\begin{bmatrix}
        Q_\mathrm C{-}\tilde Q_x\!\!&\!\!-Q_\mathrm C\tilde A_x^\top{-}H_x^\top\tilde B_y^\top \!\!&\!\! 0 \!\!&\!\! Q_\mathrm CA_x^\top+H_x^\top B_y^\top\\
        -\tilde A_xQ_\mathrm C {-}\tilde B_yH_x \!\!&\!\! U_\mathrm w \!\!&\!\! -\tilde B_uQ_\mathrm D-\tilde B_yH_u \!\!&\!\! Q_{s}B_s^\top \\
        0 \!\!&\!\! -Q_\mathrm D\tilde B_u^\top-H_u^\top\tilde B_y^\top \!\!&\!\! \tilde Q_u \!\!&\!\! Q_\mathrm DB_u^\top{+}H_u^\top B_y^\top\\
        A_xQ_\mathrm C{+}B_yH_x \!\!&\!\! B_sQ_{s} \!\!&\!\! B_uQ_\mathrm D{+}B_yH_u \!\!&\!\! Q_\mathrm C
    \end{bmatrix}\succeq 0,$}
\end{equation}
holds, where \(
U_\mathrm{w} = 2 Q_\mathrm{S} - \tilde{B}_s^{\mathrm{o}} Q_\mathrm{S} - \tilde{B}_y H_s - Q_\mathrm{S} \tilde{B}_s^{\mathrm{o}\top} - H_s^\top \tilde{B}_y^\top
\), then, setting $\theta=H\tilde Q^{-1}$, where $H=[\,H_x,\,H_u,\,H_s\,]$ and $\tilde Q=\diag(Q_\mathrm C,Q_\mathrm D,Q_{s})$, model~\eqref{eq:RNN_model_ss} is $\delta$ISS with respect to $\mathbb{R}^n$ and $\mathbb{R}^m$.
\end{proposition}
The proof of Proposition~\ref{prop:diss_opt} is provided in the Appendix for clarity reasons. Exploiting this result, we can now modify the two training approaches described in Section~\ref{sec:training_alg} to ensure that \eqref{eq:RNN_model_ss} is well-posed and $\delta$ISS.\\ 
On the one hand, as far as the least-squares approach is concerned, we need to replace \emph{Step~3} in Algorithm~\ref{alg:LS_learning} with Algorithm~\ref{alg:LS_identification_diss}. More specifically, we compute $\theta^\star$ as
\begin{equation*}\theta^\star=\text{WellPosed\_$\delta$ISS\_LS}(Y_\mathrm d, \Phi,\mathcal H, \beta).\end{equation*}
On the other hand, regarding the set-membership procedure, we need to replace \emph{Step~4} in Algorithm~\ref{alg:set_membership_learning} with Algorithm~\ref{alg:scenario_parameter_selection_diss}. In particular, we extract $\theta^{[t]}$ as 
\begin{equation*}\theta^{[t]}=\text{WellPosed\_$\delta$ISS\_scenario\_sampling}(\Theta,\mathcal H).\end{equation*}
\begin{algorithm}
\small
\caption{WellPosed\_$\delta$ISS\_LS}
\label{alg:LS_identification_diss}
\begin{algorithmic}[1]
\Require $Y_\mathrm d$, $\Phi$, $\mathcal H$.
$\tau_\mathrm w$, $\beta$
\Ensure $\theta^\star$
\State Solve
\begin{align*}
    &\min_{\substack{
        s,\lambda\in \mathbb{R_+},\,\;
        Q_\mathrm C ,\tilde Q_x \in \mathbb{S}_+^n,\\
        Q_\mathrm D,\tilde Q_u \in \mathbb{S}_+^m,\;
        Q_{s} \in \mathbb{D}_+^\nu,\\
        H_x\in\R^{p\times n},\,H_u\in\R^{p\times m}, \, H_s\in\R^{p\times \nu}
    }}
    s+\beta \lambda,\\
    &\text{subject to: }\notag\\
    &\eqref{eq:diss_opt_cost_lmi},\eqref{eq:well_posedness_lmi},\eqref{eq:lmi_training_condition_relaxed},\eqref{eq:diss_lmi}\notag
\end{align*}
\State Set
\(
\theta^\star = H \tilde Q^{-1}\), where
$H = [\,H_x,\, H_u,\, H_s\,]$ and
$\tilde Q = \operatorname{diag}(Q_\mathrm C,Q_\mathrm D, Q_{s})$.
\end{algorithmic}
\end{algorithm}
\begin{algorithm}
\small
\caption{WellPosed\_$\delta$ISS\_scenario\_sampling}
\label{alg:scenario_parameter_selection_diss}
\begin{algorithmic}[1]
\Require $\Theta$, $\mathcal H$ 
\Ensure $\tilde \theta^\mathrm s$

\State Randomly extract $\theta^{\mathrm s}$ from $(\Theta,\mathbb{P}_{\theta^{\mathrm s}})$.

\State Solve
\begin{align*}
    &\min_{\substack{\{c_i\in\R_+\}_{i=1}^p,\,
        Q_{\mathrm C},\tilde Q_x \in \mathbb{S}_+^{n},\\ Q_{\mathrm D},\tilde Q_u \in \mathbb{S}_+^{m},\,
        Q_{s} \in \mathbb{D}_+^\nu,\\
        H_x\in\mathbb{R}^{p\times n},\,H_u\in\mathbb{R}^{p\times m}
        H_s\in\mathbb{R}^{p\times \nu}}}\sum_{i=1}^{p}c_i\\
    &\text{subject to:}\notag\\
    &\eqref{eq:SM_cstr},\eqref{eq:well_posedness_lmi},\eqref{eq:diss_lmi}\notag
\end{align*}

\State  Set
\(
\theta^\star = H \tilde Q^{-1}\), where
$H = [\,H_x,\, H_u,\, H_s\,]$ and
$\tilde Q = \operatorname{diag}(Q_\mathrm C,Q_\mathrm D, Q_{s})$.
\end{algorithmic}
\end{algorithm}

\subsection{Learning structured stable models}
In this section we address the problem of imparting the $\delta$ISS during the training of structured plant models. In particular, we focus on two aspects: (i) imparting the $\delta$ISS to the single submodel $\mathcal M_i$, where $i\in\mathcal I$; (ii) imparting the $\delta$ISS to the overall structured model $\mathcal M$ obtained by interconnecting the submodels $\mathcal M_i$, for all $i\in\mathcal I$.
\subsubsection{Imparting the $\delta$ISS to the submodels}
The following proposition provides a condition for the $\delta$ISS of $\mathcal M_i$.
\begin{proposition}\label{prop:diss_submodel}
Consider the local dynamics~\eqref{eq:PI_RNN_subsystem_i} for the $i$-th submodel, where $i\in\mathcal I$ and let Assumption~\ref{ass:sigmoid_function} hold. Define the vector \(v_i(k)\coloneq[(x(k))^\top_{\mathcal N_{x,i}^\mathcal M},\,(u(k))^\top_{ \mathcal N_{u,i}^\mathcal M\cup\{i\}}]^\top\in\R^{n_{v,i}}\) and matrices  \begin{align*}
B_{v,i}^\mathrm{o} 
&\coloneqq \big[(A_x)_{\{i\},\,\mathcal N_{x,i}^{\mathcal M}},\;
               (B_u)_{\{i\},\,\mathcal N_{u,i}^{\mathcal M}\cup\{i\}}\big],\\[2mm]
\tilde B_{v,i}^\mathrm{o} 
&\coloneqq \big[(\tilde A_x)_{\{i\},\,\mathcal N_{x,i}^{\mathcal M}},\;
               (\tilde B_u)_{\{i\},\,\mathcal N_{u,i}^{\mathcal M}\cup\{i\}}\big],\\[2mm]
D_{v,i} 
&\coloneqq \big[(C)_{\{i\},\,\mathcal N_{x,i}^{\mathcal M}},\;
               (D)_{\{i\},\,\mathcal N_{u,i}^{\mathcal M}\cup\{i\}}\big].
\end{align*}
Assume that there exist matrices $H_{x,i}\in\R^{p_i\times n_i}$, $H_{v,i}\in\R^{p_i\times n_{v,i}}$, $H_{s,i}\in\R^{\nu_i\times p_i}$, $Q_{\mathrm C,i},\tilde Q_{x,i}\in\mathbb S_+^{n_i}$, $Q_{\mathrm V,i},\tilde Q_{v,i}\in\mathbb S_+^{n_{v,i}}$, and  $Q_{\mathrm S,i}\in\mathbb D_+^{\nu_i}$, such that condition~\eqref{eq:diss_lmi} is satisfied after replacing the
matrices
$A_x$, $\tilde A_x$, $B_u$, $\tilde B_u$, $B_y$, $\tilde B_y$,
$B_s$, and $\tilde B_s^{\mathrm o}$ with 
$A_{x,i}$, $\tilde A_{x,i}$, $B_{v,i}^{\mathrm o}$,
$\tilde B_{v,i}^{\mathrm o}$, $B_{y,i}$, $\tilde B_{y,i}$,
$B_{s,i}$, and $\tilde B_{s,i}^{\mathrm o}$, respectively, and
replacing the decision variables
$H_x$, $H_u$, $H_s$, $Q_{\mathrm C}$, $\tilde Q_x$, $Q_{\mathrm D}$,
$\tilde Q_u$, and $Q_{\mathrm S}$ with
$H_{x,i}$, $H_{v,i}$, $H_{s,i}$,
$Q_{\mathrm C,i}$, $\tilde Q_{x,i}$,
$Q_{\mathrm V,i}$, $\tilde Q_{v,i}$, and
$Q_{\mathrm S,i}$, respectively.
Then, setting $\theta_i=H_i\tilde Q_i^{-1}$, where $H_i=[H_{x,i}\,H_{v,i}\,H_{s,i}]$ and $\tilde Q_i=\diag(Q_{\mathrm C,i},Q_{\mathrm V,i},Q_{\mathrm S,i})$, the trained model \eqref{eq:non_overlapping_ss} of $\mathcal M_i$ is $\delta$ISS with respect to $\mathbb{R}^{n_i}$ and $\mathbb{R}^{n_{v,i}}$, i.e., where $v_i(k)$ is is regarded as the exogenous input/perturbation vector.
\hfill$\square$
\end{proposition}
The proof of Proposition~\ref{prop:diss_submodel} can be found in the Appendix. In view of it, the distributed procedure presented in Section~\ref{sec:learning structured models} can be modified so as to guarantee the well-posedness and $\delta$ISS of $\mathcal M_i$.
In particular, Algorithm~\ref{alg:distributed_identification} is modified as:
\begin{itemize}
\item \emph{Step~8} is replaced by
\begin{equation*}\theta_i^\star=\text{WellPosed\_}\delta\text{ISS\_LS}(Y_{\mathrm d,i},\Phi_i,\mathcal H_i,\beta)\end{equation*}
\item \emph{Step~17} is replaced by
\begin{equation*}\theta_i^\star=\text{WellPosed\_}\delta\text{ISS\_scenario\_sampling}(\Theta_i,\mathcal H_i)\end{equation*}
\end{itemize}
where
\(
\mathcal{H}_i \coloneqq \{ n_i, \nu_i, \sigma_i, A_{x,i}, B_{v,i}, B^o_{s,i}, B_{y,i}, \tilde{A}_{x,i}, \tilde{B}_{v,i}, \tilde{B}^o_{s,i}, \tilde{B}_{y,i} \}.
\)
\subsubsection{Imparting the $\delta$ISS to the overall plant model}
The following proposition provides a condition for the $\delta$ISS of the structured model $\mathcal M$. 
\begin{proposition}\label{prop:diss_non_overlapping}
Consider the untrained model~\eqref{eq:RNN_model}, structured according to~\eqref{eq:hyperparam_structured}–\eqref{eq:free_param_structured}, and let Assumption~\ref{ass:sigmoid_function} hold. Assume that there exist matrices
\begin{itemize}
    \item $H_x\in\mathcal C$, $H_u\in\mathcal D$, and $H_s=\diag(H_s^1,\dots,H_s^{n_s})$ where $H_s^{i}\in\mathbb R^{p_i\times \nu_i}$ for all $i\in\mathcal I$,
    
    \item
    $Q_\mathrm C=\diag(Q_\mathrm C^1,\dots,Q_\mathrm C^{n_\mathrm s})$, 
    $Q_\mathrm D=\diag(Q_\mathrm D^1,\dots,Q_\mathrm D^{n_\mathrm s})$, and
    $Q_{s}=\diag(Q_{s}^1,\dots,Q_{s}^{n_\mathrm s})$,
    where $Q_\mathrm C^i\in\mathbb S^{n_i}_+$, 
    $Q_\mathrm D^i\in\mathbb S^{m_i}_+$, 
    and $Q_{s}^i\in\mathbb D^{\nu_i}_+$ for all $i\in\mathcal I$,
\end{itemize}
such that condition~\eqref{eq:diss_lmi} holds.\\
Then, setting $\theta_i=H_i\tilde Q_i^{-1}$ for all $i\in\mathcal{I}$, where
\begin{equation}\label{eq_cstr_diss_non_overl_H}
H_{i} = \begin{bmatrix}(H_x)_{\{i\},\mathcal N_{x,i}^\mathcal M\cup\{i\}}& (H_u)_{\{i\},\,\mathcal N_{u,i}^\mathcal M\cup\{i\}}& H_s^{i}\end{bmatrix},
\end{equation}
and
\begin{equation}\label{eq_cstr_diss_non_overl_Q}
\tilde Q_{i} = \diag\left((Q_\mathrm C)_{\mathcal N_{x,i}^\mathcal M\cup\{i\},\,\mathcal N_{x,i}^\mathcal M\cup\{i\}},\,
(Q_\mathrm D)_{\mathcal N_{u,i}^\mathcal M\cup\{i\},\,\mathcal N_{u,i}^\mathcal M\cup\{i\}},Q_{s}^{i}\right),
\end{equation}
$\mathcal M$ is $\delta$ISS with respect to $\mathbb{R}^n$ and $\mathbb{R}^m$. 
\hfill$\square$
\end{proposition}
The proof of Proposition~\ref{prop:diss_non_overlapping} has been moved to the Appendix for better clarity.
To exploit this result, we define the sets
$
\mathcal Q_\mathrm C \coloneqq 
\bigl\{ Q_\mathrm C \in \mathbb S_+^n : 
Q_\mathrm C = \diag(Q_{\mathrm C}^1,\dots,Q_{\mathrm C}^{n_\mathrm s}), \, 
Q_{\mathrm C}^i\in\mathbb S_+^{n_i}
\bigr\}$ and
$\mathcal Q_\mathrm D \coloneqq 
\bigl\{ Q_\mathrm D \in \mathbb S_+^m : 
Q_\mathrm D = \diag(Q_{\mathrm D}^1,\dots,Q_{\mathrm D}^{n_\mathrm s}), \, 
Q_{\mathrm D}^i\in\mathbb S_+^{m_i}
\bigr\}$.\\
Algorithm~\ref{alg:distributed_identification} is modified as follows:
\begin{itemize}
    \item \emph{Step~8} is removed, and the following centralised problem is solved
\begin{subequations}\label{opt:LS_identification_structured_diss}
{\small
\begin{align}
    &\min_{\substack{
        s_i,\lambda_i\in \mathbb{R_+},\,\forall i\in\mathcal I\\
        \tilde Q_x\in\mathbb S_+^{n},\tilde Q_u\in\mathbb S_+^{m}\\
        Q_\mathrm C\in\mathcal Q_\mathrm C,Q_\mathrm D\in\mathcal Q_\mathrm D,Q_{s}\in\mathbb D_+^{\nu}\\
        H_x\in\mathcal C,H_u\in\mathcal D,H_s\in\mathcal D_s
    }}
    \sum_{i=1}^{n_\mathrm s} \left(s_i+\beta \lambda_i\right)\,,\\
    &\text{subject to:}\notag\\
    &\eqref{eq:diss_lmi},\notag\\
    &\forall i\in\mathcal I:\notag\\
    &\quad \begin{bmatrix}
        s_i \!\!\!&\!\!\!  H_i \!-\! (Q_{\mathrm d,i}^{-1}\Phi_i^\top Y_{\mathrm d,i})^\top\tilde Q_i\\
        H_i^\top \!-\! \tilde Q_i {Q_{\mathrm d,i}}^{-1}\Phi_i^\top Y_{\mathrm d,i} \!\!\!&\!\!\! \tilde Q_i
    \end{bmatrix}\succeq 0,\notag\\
    &\quad \tilde Q_i-\frac{1}{N_\mathrm d-\tau_\mathrm w}Q_{\mathrm d,i}+\lambda_i I_{r_i}\succeq0,\notag\\
    &\quad -\tilde Q_i+\frac{1}{N_\mathrm d-\tau_\mathrm w}Q_{\mathrm d,i}+\lambda_i I_{r_i}\succeq0,\notag\\
    &\quad 2Q_{s}^i{-}\tilde B_{s,i}^\mathrm oQ_{s}^i-\tilde B_{y,i}H_s^i{-}Q_{s}^i\tilde{B}_{s,i}^{\mathrm{o}\,\top}{-}H_{s}^{i^\top}\tilde B_{y,i}^\top{\succ}0\label{eq:well_posedness_lmi_submodel}
\end{align}
}
\end{subequations}
where $H_i$ and $\tilde Q_i$ are computed according to~\eqref{eq_cstr_diss_non_overl_H} and  \eqref{eq_cstr_diss_non_overl_Q}. The local parameters are then computed by setting $\theta^\star_i=H_i\tilde Q_i^{-1}$, for all $i\in\mathcal I$.
\item \emph{Steps~16--18} are removed. A sample $\tilde \theta_i^{[t]}$ is randomly extracted from $(\Theta_i,\mathbb P_{\tilde\theta_i^{[t]}})$ for all $i\in\mathcal I$, and the following centralised optimisation problem is solved
\begin{subequations}\label{opt:scenario_sampling_structured_diss}
{\small
\begin{align}
    &\min_{\substack{
        \{c_j^i\in\R_+\}_{j=1}^{p_i},\,\forall i\in\mathcal I\\
        \tilde Q_x\in\mathbb S_+^{n},\tilde Q_u\in\mathbb S_+^{m}\\
        Q_\mathrm C\in\mathcal Q_\mathrm C,Q_\mathrm D\in\mathcal Q_\mathrm D,Q_{s}\in\mathbb D_+^{\nu}\\
        H_x\in\mathcal C,H_u\in\mathcal D,H_s\in\mathcal D_s
    }}
    \sum_{i=1}^{n_\mathrm s}\sum_{j=1}^{p_i} c_j^i,\\
    &\text{subject to:}\notag\\
    &\eqref{eq:diss_lmi},\notag\\
    &\forall i\in\mathcal I:\notag\\
    &\ \ \begin{bmatrix}
        c_j^i \!\!\!&\!\!\! {\theta_i^{[t]}}^{(j)}\tilde Q_i {-} H_i^{(j)}\\
        ({\theta^{[t]}}^{(j)}\tilde Q_i {-} H_i^{(j)})^\top \!\!\!&\!\!\!\tilde Q_i
    \end{bmatrix}\succeq0,\forall j{=}1,{\dots},p_i \notag\\
    &\ \ \eqref{eq:well_posedness_lmi_submodel}\notag
\end{align}
}
\end{subequations}
where $H_i$ and $\tilde Q_i$ are computed according to~\eqref{eq_cstr_diss_non_overl_H} and  \eqref{eq_cstr_diss_non_overl_Q}. We finally set $\theta_i^{[t]}=H_i\tilde Q_i^{-1}$, for all $i\in\mathcal I$.
\end{itemize}
Note that the optimisation problems~\eqref{opt:LS_identification_structured_diss} and~\eqref{opt:scenario_sampling_structured_diss} are formulated as sums of local cost functions subject to local constraints and a single global coupling constraint given by the $\delta$ISS constraint~\eqref{eq:diss_lmi}, which makes the problems centralised.
Nevertheless, owing to the modular structure of the involved matrices, the number of decision variables (and hence the computational complexity) of the problem is lower than in the case of unstructured plants. Moreover, the initialisation and simulation of the RNN reservoirs in Algorithm~\ref{alg:distributed_identification} are performed independently by each subplant. To improve scalability, future work will be devoted to the parallelisation of~\eqref{opt:LS_identification_structured_diss} and~\eqref{opt:scenario_sampling_structured_diss}, e.g., along the lines of~\cite{giselsson2013accelerated}.
%
%
%
%
%
%
\section{Simulations}
\label{sec:SIMULATIONS}
In this section the proposed physics-informed learning framework is validated through three case studies: the training of the structured model of the chemical plant described, e.g., in~\cite{stewart2011cooperative}, the learning of the $\delta$ISS model of the pH-neutralisation process previously considered, e.g., in~\cite{henson2002adaptive}, and the structured learning of an irrigation network described, e.g., in~\cite{Cantoni_irrigation}.
\subsection{Learning the structured model of a chemical plant}\label{subsec:chemical_plant}
In this section we apply Algorithm~\ref{alg:distributed_identification} to the data drawn from the chemical plant described in~\cite{stewart2011cooperative}.
The plant consists of two reactors and a separator, where a reactant~A is converted into the desired product~B and an undesired side product~C. 
The resulting physical model is a nonlinear process consisting of $n=12$ states and $m=6$ inputs. We refer the reader to~\cite{stewart2011cooperative} for a detailed description of the model.\\
As discussed in \cite{stewart2011cooperative}, the plant exhibits a modular structure. In particular, it consists of $n_\mathrm{s}=3$ interacting subplants. For each subplant $\mathcal{P}_i$, where~$i \in \{1,2,3\}$, the vector of measured variables is
\(
y_i = [\,H_i,\, x_{\mathrm{A}i},\, x_{\mathrm{B}i},\, T_i\,]^\top,
\)
where $H_i$, $T_i$, $x_{\mathrm{A}i}$, and $x_{\mathrm{B}i}$ denote the mixture level, temperature, and concentrations of reactants A and B, respectively.
The control input vectors are defined for $i=1,2$ as
\(
u_i = [\,F_{\mathrm{f}i},\, Q_i\,]^\top,
\) where $F_{\mathrm{f}i}$ is the inlet flow rate of reactant~A and $Q_i$ the external heat input, whereas, for subplant~$3$,
\(
u_3 = [\,F_\mathrm R,\, Q_3\,]^\top,
\)
where $F_{\mathrm R}$ is the recycle flow rate and $Q_3$ the heat input to the separator.

Based on the plant modular structure, the RNN model has been structured into three interconnected submodels with $(n_1,\,\nu_1)=(12,\,4)$, $(n_2,\,\nu_2)=(22,\,4)$, $(n_3,\,\nu_3)=(21,\,5)$, and $\sigma_i^{(j)}=\tanh(\cdot)$ for all $i\in\mathcal I$ and $j=1,\dots,\nu_i$,  where $\tanh(\cdot)$ denotes the hyperbolic tangent function. Moreover, the following neighbouring sets have been defined:
$\mathcal N_{x,1}^\mathcal M = \mathcal N_{u,1}^\mathcal M = \{3\}$, 
$\mathcal N_{x,2}^\mathcal M = \mathcal N_{u,2}^\mathcal M = \{1\}$, and 
$\mathcal N_{x,3}^\mathcal M = \mathcal N_{u,3}^\mathcal M = \{2\}$.
\smallskip\\
Three independent datasets have been collected with a sampling time $T_\mathrm{s}=0.1$~s: a training dataset of length $N_\mathrm d=8000$, and validation and test datasets of length $N_\mathrm v=N_\mathrm t=4000$, respectively. Each dataset has been generated by feeding the simulator based on the physical equations of the plant with multilevel pseudo-random signals (MPRS) designed to excite the plant over different operating frequencies and regions. Bounded additive white noise has been introduced in the final measurements to account for measurement uncertainty. Finally, the data have been normalised so that each variable lies within the interval $[0,\,1]$.\\
Assuming that each subplant has access only to its local measurements, a model consisting of three interconnected submodels has been trained using Algorithm~\ref{alg:distributed_identification}. In particular, two plant models have been derived: one based on the least-squares approach and the other on set-membership. The learning process is distributed and requires the transmission of the sequences $(U_{\mathrm d,3} ,X_{\mathrm d,3})$ to subplant 1, $(U_{\mathrm d,1} ,X_{\mathrm d,1})$ to subplant 2, and $(U_{\mathrm d,2} ,X_{\mathrm d,2})$  to subplant 3, while the output measurements $\mathcal Y_{\mathrm d,i}$ for all $i\in\mathcal I$ remain local.\\
To assess the effectiveness of the proposed approach, our results have been compared with those obtained with the physics-informed framework proposed in~\cite{bonassi2022recurrent}, implemented using three interconnected gated recurrent unit (GRU) models (each consisting of a single layer with 32 units), and trained using Truncated Back-Propagation Through Time (TBPTT). Unlike the proposed approach, the training procedure in~\cite{bonassi2022recurrent} is centralised and requires each subplant $i\in\mathcal I$ to share its local datasets $(\mathcal U_{\mathrm d,i},\mathcal Y_{\mathrm d,i})$ with a central computational unit. Therefore, this approach cannot be applied when network constraints prevent the subplants from sharing their local measurements.
The physics-informed GRU model has been implemented using PyTorch 2.7.1 and trained using an early stopping criterion. All experiments have been performed on a PC equipped with an AMD Ryzen 7 7435HS CPU and running Microsoft Windows 11.\\ 
To evaluate the performance of the so-obtained models, the FIT~[\%] index is computed for each output as defined in~\cite{bonassi2022recurrent}.
As shown by the FIT indices reported in Table~\ref{tab:table_fit}, the proposed RNN learned using both least-squares and set-membership approaches achieves higher modelling accuracy than the physics-informed GRU model. Moreover, the set-membership approach achieves the highest average FIT index. These results can also be visually inspected in Figure~\ref{fig:reactor2}, where, for compactness, only the modelling performance for submodel~$\mathcal M_2$ on the test dataset is reported.\\
From a computational perspective, the proposed approach is distributed, and each subplant requires, on average, $0.2$~s for least-squares and $13$~s for set-membership to train its local submodel. In contrast, the training procedure in~\cite{bonassi2022recurrent} is centralised and requires access to the output measurements of all subplants. In particular, the physics-informed GRU requires approximately $30$~min of training time, making it significantly more computationally demanding.
\begin{table}[t]
\centering
\begin{tabular}{c c c c}
\hline
\textbf{Output} &
\textbf{LS [\%]} &
\textbf{SM [\%]} &
\textbf{GRU [\%]} \\
\hline
$H_1$     & $97.83$ & $97.80$ & $84.38$ \\
$x_{A1}$  & $85.65$ & $85.46$ & $85.24$ \\
$x_{B1}$  & $87.60$ & $86.12$ & $86.03$ \\
$T_1$     & $81.93$ & $86.61$ & $78.71$ \\
$H_2$     & $96.12$ & $97.02$ & $90.66$ \\
$x_{A2}$  & $87.12$ & $87.33$ & $80.59$ \\
$x_{B2}$  & $87.93$ & $87.28$ & $80.79$ \\
$T_2$     & $81.18$ & $86.46$ & $77.87$ \\
$H_3$     & $66.72$ & $65.34$ & $25.89$ \\
$x_{A3}$  & $82.10$ & $82.31$ & $81.17$ \\
$x_{B3}$  & $74.54$ & $76.81$ & $72.37$ \\
$T_3$     & $79.66$ & $83.49$ & $61.55$ \\
\hline\hline
\textbf{Average} & \textbf{84.03} & \textbf{85.17} & \textbf{75.44} \\
\hline
\end{tabular}
\caption{FIT indices: proposed RNN learned using least-squares (LS) and set-membership (SM), and physics-informed GRU model.}
\label{tab:table_fit}
\end{table}
\begin{figure}
    \centering
    \begin{minipage}{0.49\columnwidth}
        \centering
        \includegraphics[width=\linewidth]{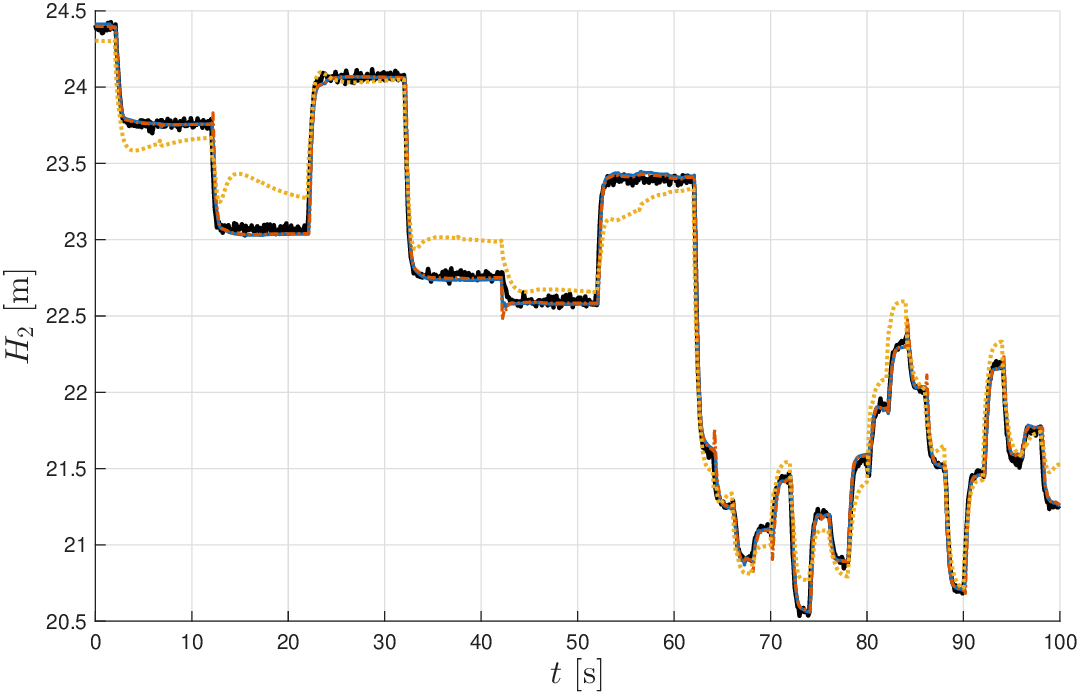}
    \end{minipage}\hfill
    \begin{minipage}{0.49\columnwidth}
        \centering
        \includegraphics[width=\linewidth]{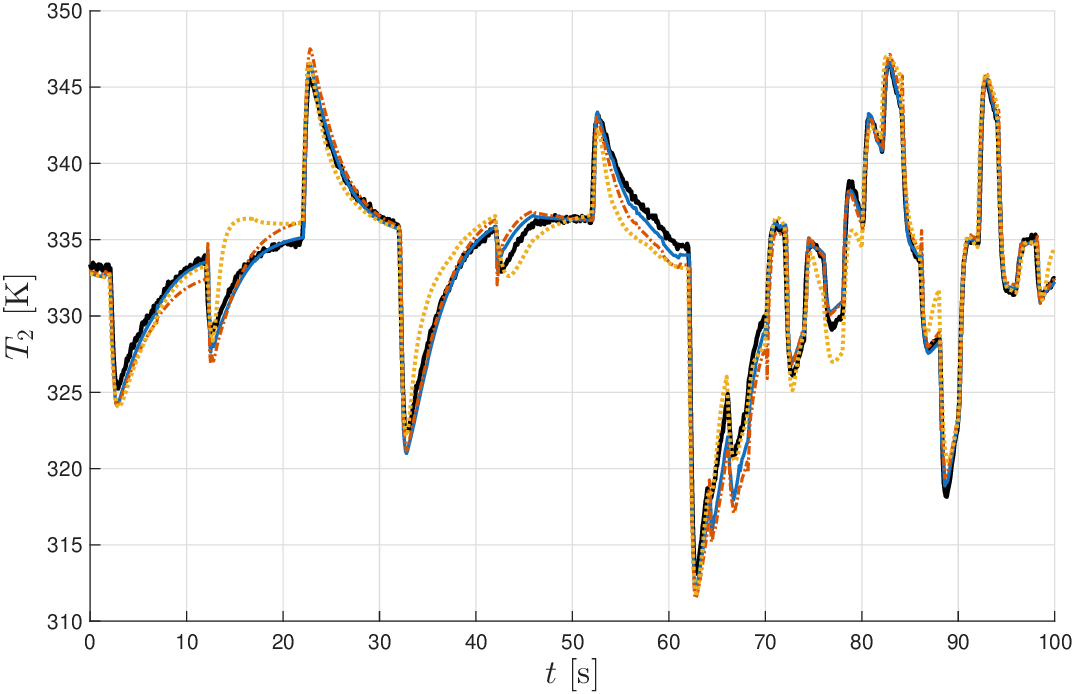}
    \end{minipage}

    \vspace{0.2cm}
    \begin{minipage}{0.49\columnwidth}
        \centering
        \includegraphics[width=\linewidth]{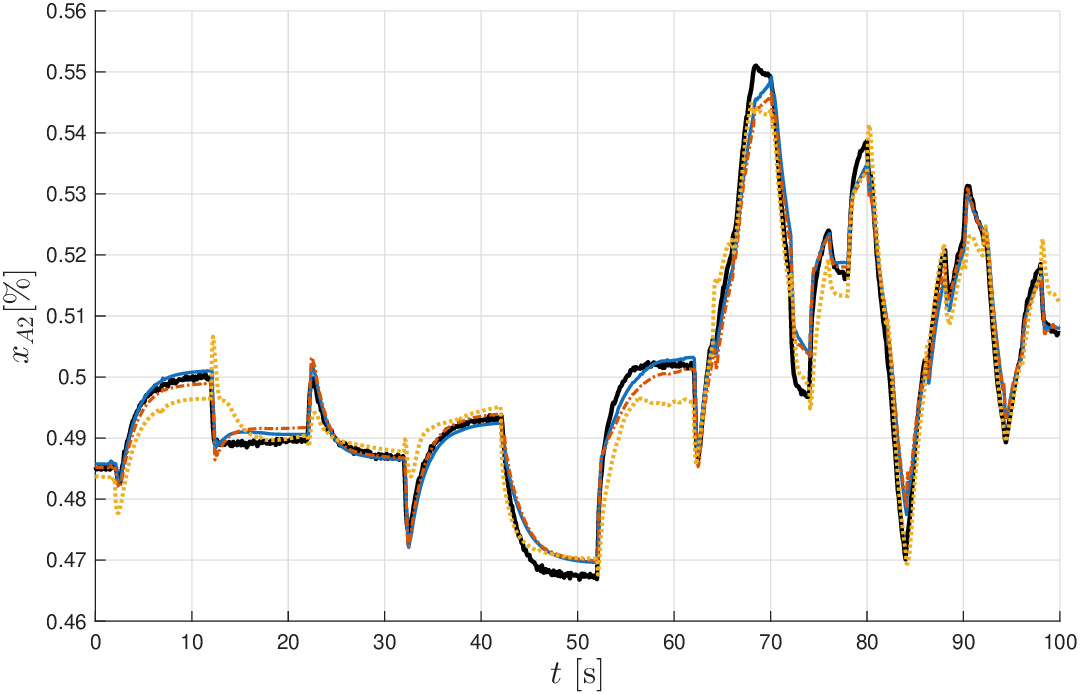}
    \end{minipage}\hfill
    \begin{minipage}{0.49\columnwidth}
        \centering
        \includegraphics[width=\linewidth]{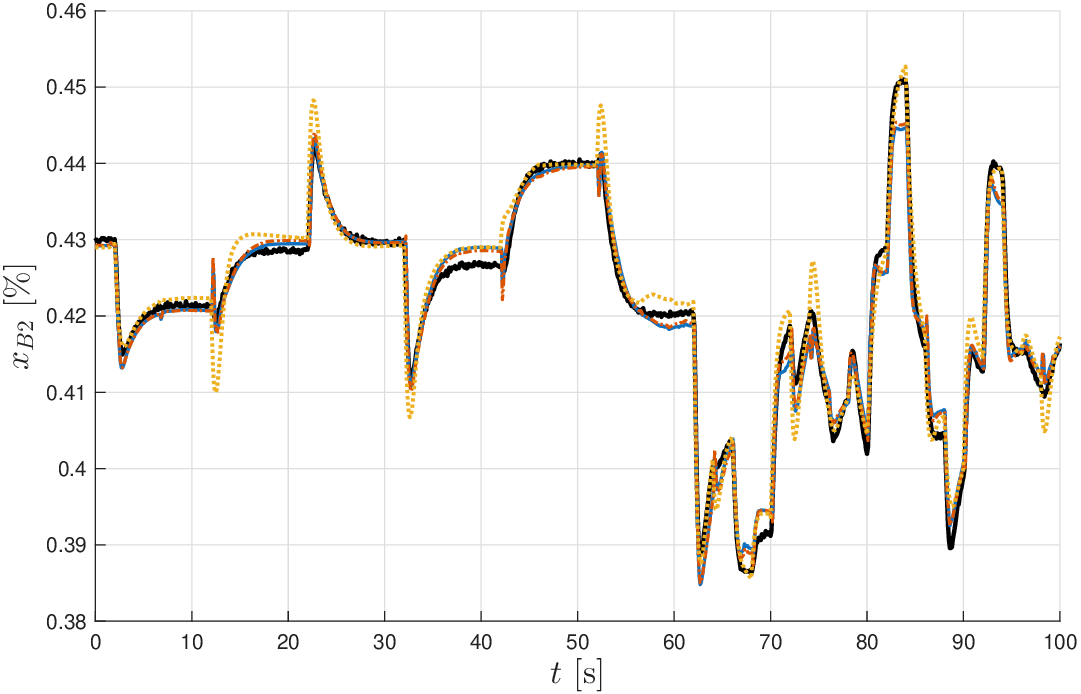}
    \end{minipage}

    \caption{Reactor 2 modelling: predictions of the proposed RNN learned using least-squares (orange) and set-membership (blue) approaches, and of the physics-informed GRU model (yellow), with the ground truth (black).}
    \label{fig:reactor2}
\end{figure}
\subsection{Learning the $\delta$ISS model of a pH-neutralisation process}
In this section we consider the case of the pH-neutralisation process described in~\cite{henson2002adaptive}.
The physical model of the process, described in detail in~\cite{henson2002adaptive}, is a continuous-time SISO plant. The input $u$ is the inlet alkaline-base flow rate, while the measured output $y$ is the pH of the outlet flow rate.\smallskip\\
A simulator based on the plant physical equations has been implemented in MATLAB for data generation: two independent sequences of length $N_\mathrm d=5000$ and $N_\mathrm t=1500$, used for training and testing, respectively, have been collected with a sampling time $T_\mathrm s=15$~s by exciting the simulator with an MPRS. Additive white noise has subsequently been added to the final measurements, and the data have been normalised so that both inputs and outputs lie within $[0,\,1]$.
\\
An RNN model with $n=20$ and $\sigma^{(i)}=\tanh(\cdot)$, for $i=1,\dots,14$, satisfying the $\delta$ISS property, has been obtained using Algorithm~\ref{alg:LS_learning}, following the procedure described in Section~\ref{sec:diss_identification}.\\
For comparison, we have trained a GRU network with three layers of $6$ units each and a long short-term memory (LSTM) network with three layers of $8$ units each on the same dataset using the TBPTT algorithm. The training procedures described in \cite{bonassi2024nonlinear} and \cite{bonassi2022recurrent} have been adopted to enforce the $\delta$ISS property of the resulting models. The same computational setup described in Section~\ref{subsec:chemical_plant} has been used for these experiments.
The modelling performance of the three models on the test dataset is reported in Figure~\ref{fig:diss_model}, where the proposed architecture is denoted by REN. All models achieve high prediction accuracy, with FIT indices of $92.1\%$, $93.5\%$, and $93.9\%$ for the REN, GRU, and LSTM models, respectively. While the GRU and LSTM models yield slightly higher FIT indices, they come at the cost of a substantially higher computational burden. In particular, training the LSTM and GRU models requires approximately $75$~min and $50$~min, respectively, whereas the proposed method identifies a $\delta$ISS model in only $1.1$~s.
\begin{figure}
     \centering
     \includegraphics[width=\linewidth]{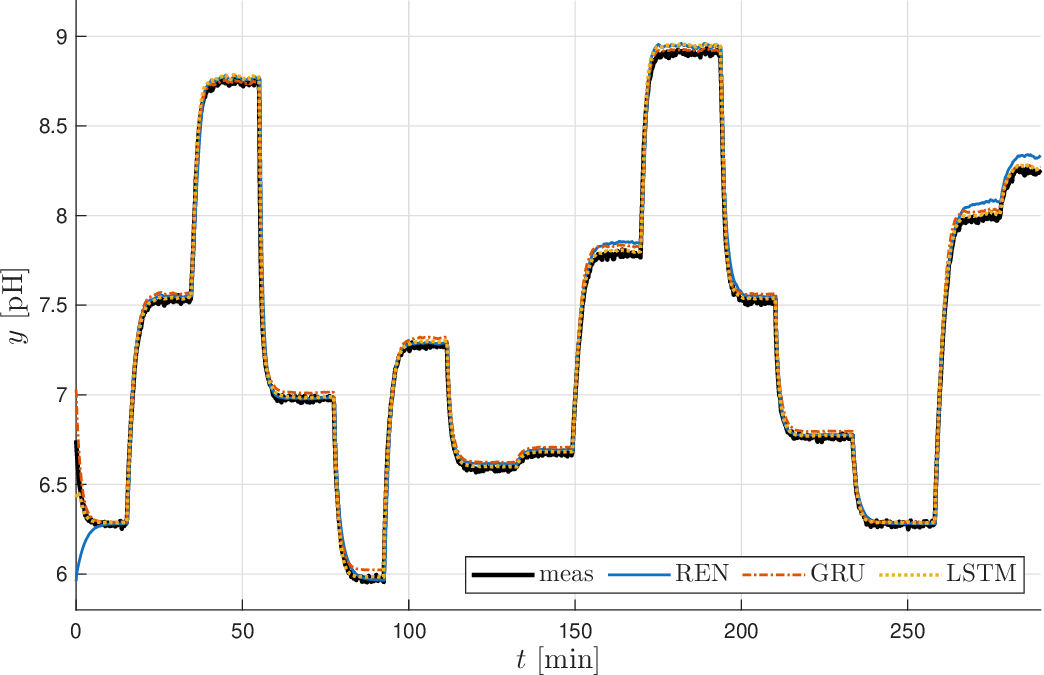}
    \caption{pH-neutralisation modelling: comparison of the proposed RNN model (REN) with GRU and LSTM models.}
      \label{fig:diss_model}
\end{figure}
\subsection{Learning a structured model of an irrigation network}
In this section, we apply Algorithm~\ref{alg:distributed_identification} to learn a model of an irrigation network, i.e., an infrastructure comprising reservoirs and open-water channels that transport water under the power of gravity. Each channel is typically modelled as a sequence of pools interconnected by gates, whose heights can be adjusted to control the flow between consecutive pools.\\
In our framework, each pool is considered a subsystem.  For each pool $i$, the measurable water level is denoted by $h_i$, while the height of the gate connecting pools $i$ and $j$ is denoted by $q_{ij}$. Whenever $q_{ij}\neq0$, the dynamics of pool $j$ is affected by $h_i$ and $q_{ij}$ with a delay $\tau_i\in\mathbb Z_+$, accounting for the water transport time. The dynamics of each pool is described by the continuous-time nonlinear model described in~\cite{Cantoni_irrigation}. Although the continuous-time model of each pool is first-order, its discrete-time representation requires a state dimension that grows with  $\tau_i/T_\mathrm s$ in order to account for water transportation delays.
The irrigation network considered in this study is schematically illustrated in Figure~\ref{fig:irrigation_plant_structure} and comprises $n_s=15$ interconnected pools. For each pool $i$, the sets $\mathcal N_{(u,i)}^\mathcal M$ and $\mathcal N_{(x,i)}^\mathcal M$ are determined by the interconnection structure shown in Figure~\ref{fig:irrigation_plant_structure}, with $\mathcal N_{(u,i)}^\mathcal M=\mathcal N_{(x,i)}^\mathcal M$. The input vector is defined as $u_i=[q_{ij} ]_{j\in\{r\in\mathcal I:i\in\mathcal N_{(u,r)}^{\mathcal M} \}}$ and the output as $y_i=h_i$.\\
The first-principles model of the system with $\tau_i=16~\mathrm{min}$ for all $i\in\mathcal I$ has been implemented in the MATLAB environment. Independent sequences of length $N_\mathrm d=7000$, $N_\mathrm v=2000$, and $N_\mathrm t=1500$ have been extracted with a sampling time $T_\mathrm s=4~\mathrm{min}$ by exciting the simulator with MPRS and used as training, validation, and test datasets.\\
Given the large spatial scale of irrigation networks, we assume that each subsystem performs local measurements and has access only to its own dataset. An RNN model replicating the plant interconnection structure, with $n=246$ and $\sigma_r (\cdot)=\mathrm{tanh}(\cdot)$ for $r=1,\dots, 89$, has been obtained in a distributed fashion using Algorithm~\ref{alg:distributed_identification}. The learning procedure takes an average of $0.43~\mathrm{s}$ per subplant for the least-squares approach and $ 37.26~\mathrm{s}$ per subplant for the set-membership approach, requiring communication only between neighbouring subsystems. The resulting models achieve average $\mathrm{FIT}$ of  $85.60\%$ and $85.90\%$ for the least-squares and set-membership approaches, respectively, on the test dataset, indicating good modelling accuracy. Figure~\ref{fig:irrigation_network_plots} illustrates the model performance for subsystems $4$, $9$, $12$, and $15$.\\
It is worth noting that obtaining an equivalent model through a centralised learning approach (either using Algorithm~\ref{alg:LS_learning} or~\ref{alg:set_membership_learning}), with the same state dimension and number of nonlinearities, would not be computationally viable due to numerical issues. Moreover, centralised training would require collecting the datasets from all subplants at a central unit, thus violating the assumed data-accessibility constraints. This further highlights the scalability advantages of the proposed distributed learning approach. 
\begin{figure}
 \centering
   \includegraphics[width=0.7\linewidth]{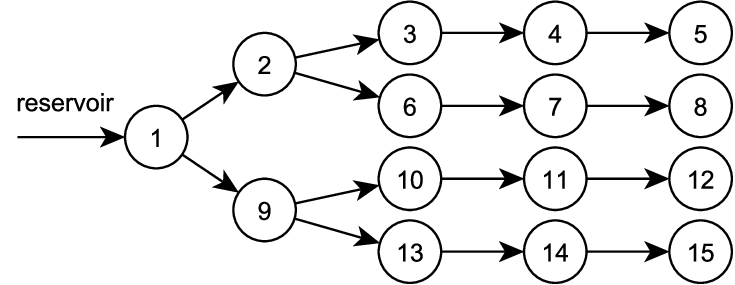}
   \caption{Irrigation network block structure.}\label{fig:irrigation_plant_structure}
\end{figure}
\begin{figure}
    \centering
    \begin{minipage}{0.49\columnwidth}
        \centering
        \includegraphics[width=\linewidth]{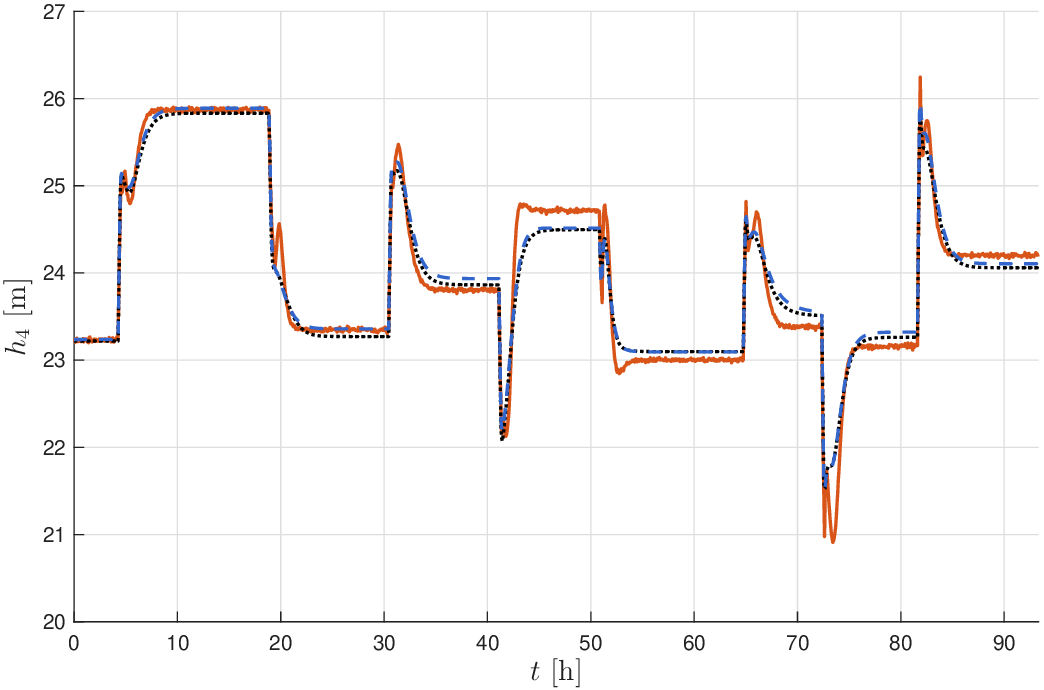}
    \end{minipage}\hfill
    \begin{minipage}{0.49\columnwidth}
        \centering
        \includegraphics[width=\linewidth]{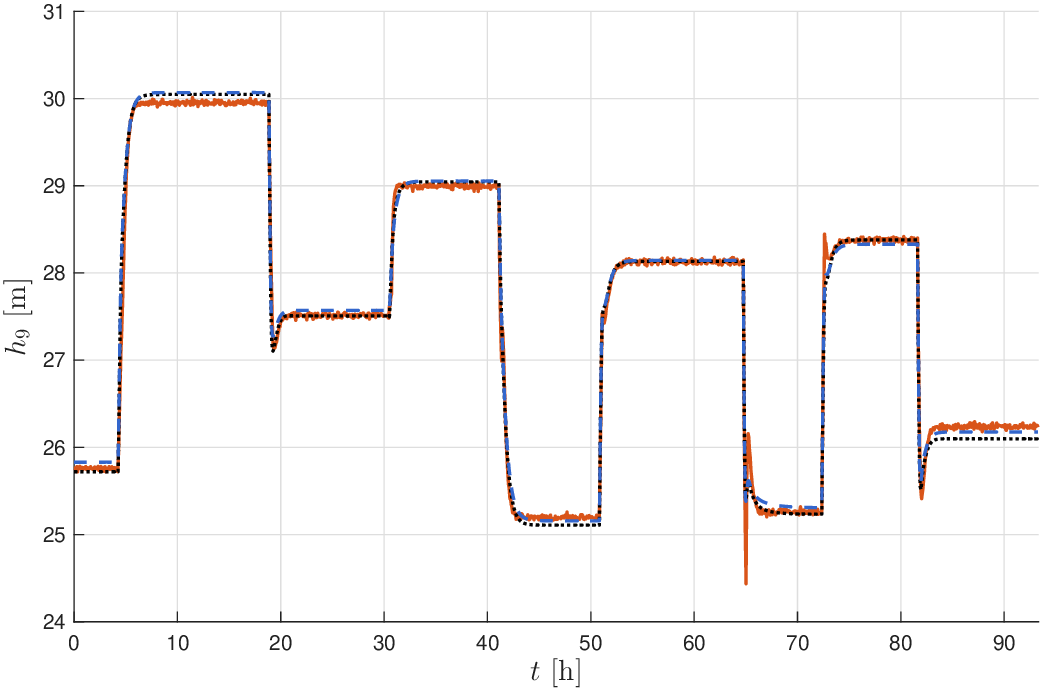}
    \end{minipage}

    \vspace{0.2cm}
    \begin{minipage}{0.49\columnwidth}
        \centering
        \includegraphics[width=\linewidth]{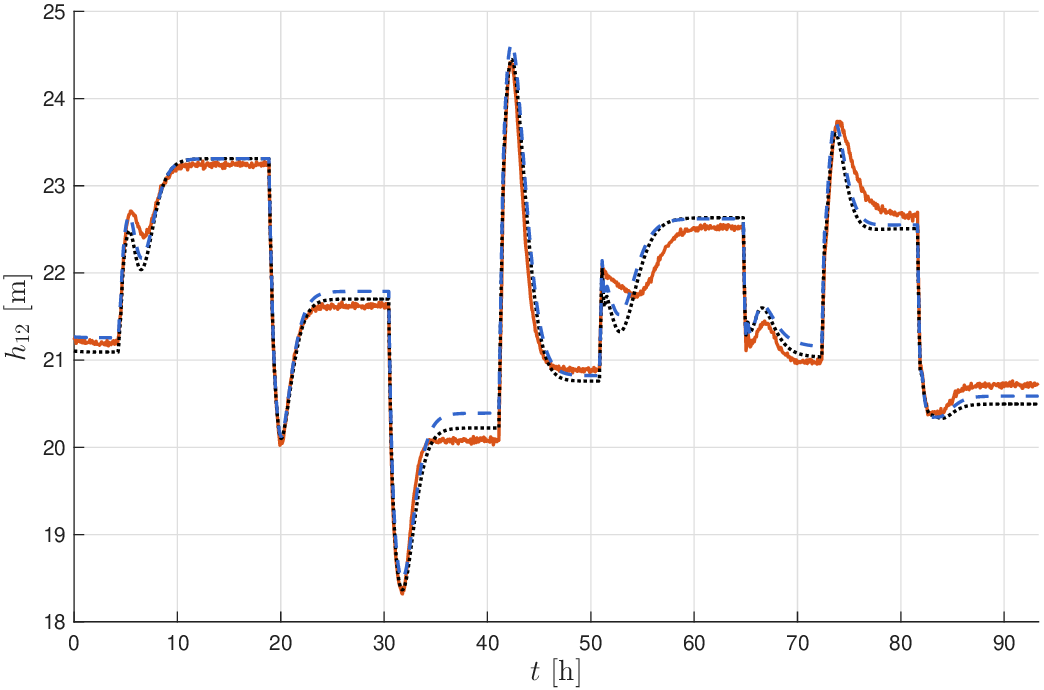}
    \end{minipage}\hfill
    \begin{minipage}{0.49\columnwidth}
        \centering
        \includegraphics[width=\linewidth]{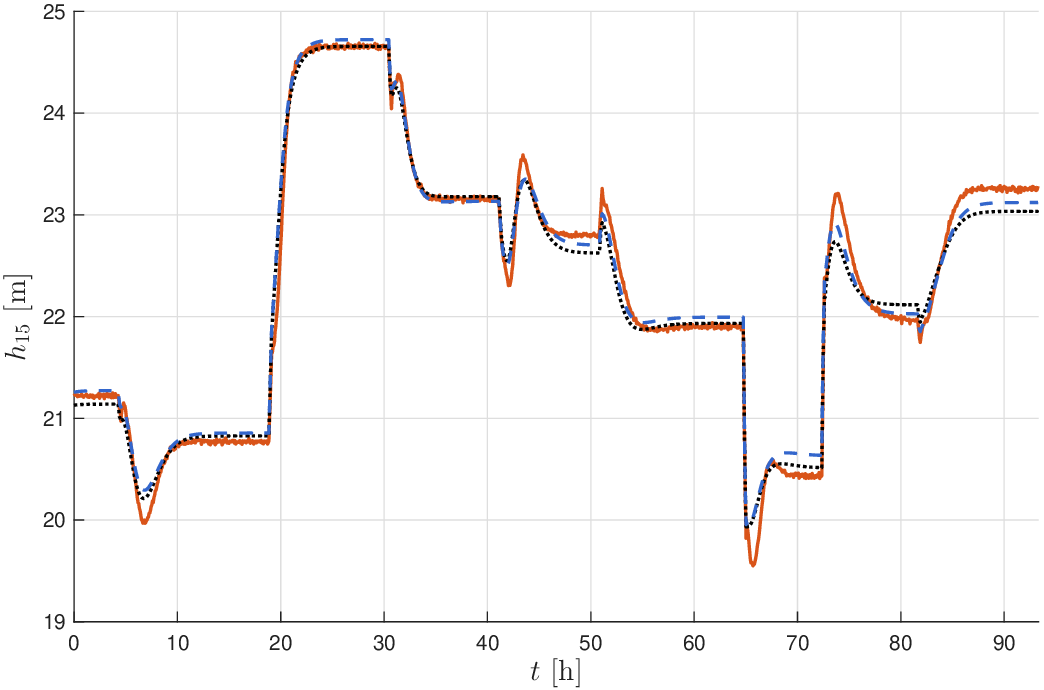}
    \end{minipage}

    \caption{Irrigation network modelling: predictions of the proposed RNN learned using least-squares (dashed blue) and set-membership (dotted black) approaches, compared with the ground truth (solid orange).}
    \label{fig:irrigation_network_plots}
\end{figure}
%
%
\section{Conclusions}
\label{sec:conclusions}
In this paper, we presented a framework for physics-informed learning of a class of RNNs that preserve the plant's structural and stability properties.
The learning problem is formulated as a convex optimisation one, enabling the inclusion of LMI constraints. 
Notably, when the plant exhibits a modular structure, the training can be performed in a distributed manner, making the approach well-suited to networked plants.\\
Future work will focus on extending the proposed framework to plants exhibiting different qualitative properties, e.g., chaotic behaviours and integrator dynamics, and on assessing its performance on very large-scale realistic ones.
%
\appendix
\section{Proof of the main results}
In this appendix, we report the proofs of the main results of the paper.\smallskip\\
Before proceeding, we state the following lemma, 
which gives a sufficient condition to ensure that model~\eqref{eq:RNN_model_ss} is $\delta$ISS.
\begin{lemma}\label{th:dISS_condition}
    Consider the dynamics~\eqref{eq:RNN_model_ss}, and let Assumption~\ref{ass:sigmoid_function} hold. If there exist matrices $P,Q_x\in\mathbb S_+^n$, $Q_u\in\mathbb S_+^m$, and $\Lambda\in \mathbb{D}_+^\nu$ such that
    \begin{equation}
        \begin{bmatrix}
            P-Q_x \!&\! -\tilde{A}^\top\Lambda \!&\! 0 \!&\! A^\top P\\ -\Lambda\tilde{A} \!&\! 2\Lambda-\Lambda\tilde{B_s}-\tilde{B}_s^\top\Lambda \!&\! -\Lambda \tilde{B} \!&\! B_s^\top P\\
            0\!&\!-\tilde B^\top \Lambda \!&\! Q_u \!&\! B^\top P \\
            PA\!&\!PB_s\!&\!PB\!&\!P
        \end{bmatrix}\succeq0,
        \label{eq:dISS_suff_cond}
    \end{equation}
   then, \eqref{eq:RNN_model_ss} is $\delta$ISS.
\end{lemma}
\textbf{Proof of Lemma~\ref{th:dISS_condition}}
In order to prove  the $\delta$ISS of~\eqref{eq:RNN_model_ss}, we show
the existence of a dissipation-form $\delta$ISS Lyapunov function. In particular, we
consider, as a candidate,  \(V(x_1(k),x_2(k))\coloneq\norm{x_1(k)-x_2(k)}_{P}^2,\) and we show
that condition~\eqref{eq:dISS_suff_cond} implies
\begin{equation}\label{eq:lyapunov_dISS}
\Delta V\leq-\norm{\Delta x(k)}^2_{Q_x}+\norm{\Delta u(k)}^2_{Q_u}.
\end{equation}
where $\Delta V\coloneq V(x_1(k+1),x_2(k+1))-V(x_1(k),x_2(k))$, $\Delta x(k)=x_1(k)-x_2(k)$, and $\Delta u(k)=u_1(k)-u_2(k)$.\\
For each $j=1,2$, consider $x_j(k)\in\R^n$ and $u_j(k)\in\R^m$, and denote $v_j(k)\coloneq\tilde A x_j(k)+\tilde Bu_j(k)+\tilde B_ss_j(k)$, $s_j\coloneq\sigma(v_j)$, and $x_j(k+1)\coloneq Ax_j(k)+Bu_j(k)+B_ss_j(k)$. 
The dynamcis of $\Delta x$ is
\begin{equation}\label{eq:error_dynamics_diss}
\Delta x(k+1)=A\Delta x(k)+B\Delta u(k)+B_s\Delta s(k),
\end{equation}
where $\Delta s(k)\coloneq s_1(k)-s_2(k)$.\\ 
Under Assumption~\ref{ass:sigmoid_function}, in view of~\cite[Lemma~2]{ravasio2025developmentvelocityformclass}, it holds that 
\begin{equation}\label{eq:sector_inequality} (\Delta v(k)-\Delta s(k))^\top \Lambda\Delta s(k)\geq 0,
\end{equation}
for any $\Lambda\in\mathbb D_+^\nu$, where $\Delta v(k)\coloneq v_1(k)-v_2(k)$.
Noting that 
\(\Delta v(k)=\tilde A \Delta x(k)+\tilde B\Delta u(k)+\tilde B_s\Delta s(k),\)
condition~\eqref{eq:sector_inequality} is equivalent to
\begin{equation}\label{eq:sector_inequality_ext} (\tilde A \Delta x(k)+\tilde B\Delta u(k)+(\tilde B_s-I_\nu)\Delta s(k))^\top \Lambda\Delta s(k)\geq 0,
\end{equation}
Defining $\phi_\mathrm s=[\Delta x(k)^\top,\,\Delta s(k)^\top,\,\Delta u(k)^\top]$, condition~\eqref{eq:sector_inequality_ext} implies that
\begin{multline*}
    \phi_\mathrm s^\top\begin{bmatrix}\tilde A^\top\\\tilde B_{s}^\top-I_\nu\\\tilde B^\top\end{bmatrix}\Lambda\begin{bmatrix}0&I_\nu&0\end{bmatrix}\phi_\mathrm s \\+
    \phi_\mathrm s^\top\begin{bmatrix}0\\I_\nu\\0\end{bmatrix}\Lambda\begin{bmatrix}\tilde A&\tilde B_s-I_\nu&\tilde B\end{bmatrix}\phi_\mathrm s\geq 0,
\end{multline*}
which is equivalent to
\begin{equation}\label{eq:sector_inequality_quadr}
\phi_\mathrm s\begin{bmatrix}0&\tilde A\Lambda&0\\\Lambda\tilde A&-2\Lambda +\tilde B_s^\top\Lambda+\Lambda\tilde B_s&\Lambda\tilde B\\
0&\tilde B\Lambda&0\end{bmatrix}\phi_\mathrm s\geq 0\,.
\end{equation}
We can exploit~\eqref{eq:sector_inequality_quadr} to guarantee~\eqref{eq:lyapunov_dISS}, by imposing
\begin{multline}\label{eq:diss_sec_cond}
    \Delta V+\phi_\mathrm s^\top\begin{bmatrix}0&\tilde A\Lambda&0\\\Lambda\tilde A&-2\Lambda +\tilde B_s^\top\Lambda+\Lambda\tilde B_s&\Lambda\tilde B\\
0&\tilde B\Lambda&0\end{bmatrix}\phi_\mathrm s        \\\leq-\norm{\Delta x(k)}^2_{Q_x}+\norm{\Delta u(k)}^2_{Q_u},
\end{multline}
Using~\eqref{eq:error_dynamics_diss}, it follows that
\begin{align*}
\Delta V&=(A\Delta x(k)+B\Delta u(k)+B_s\Delta s(k))^\top P(A\Delta x(k)+\\&\quad B\Delta u(k)+B_s\Delta s(k))-\Delta x(k)^\top P\Delta x(k)\\
&=\phi^\top\left(\begin{bmatrix}
    A^\top\\B_s^\top\\B^\top
\end{bmatrix}P\begin{bmatrix}
    A&B_s&B
\end{bmatrix}-\begin{bmatrix}
    P&0&0\\0&0&0\\0&0&0\end{bmatrix}\right)\phi.
\end{align*}
Substituting the derived expression for $\Delta V$, condition~\eqref{eq:diss_sec_cond} is equivalent to
\begin{multline}\label{eq:diss_condition_matrix_form}
\phi_\mathrm s^\top\Bigg(\begin{bmatrix}
    A^\top\\B_s^\top\\B^\top
\end{bmatrix}P\begin{bmatrix}
    A&B_s&B
\end{bmatrix}\\-\begin{bmatrix}P-Q_x&-\tilde A\Lambda&0\\-\Lambda\tilde A&2\Lambda -\tilde B_s^\top\Lambda-\Lambda\tilde B_s&-\Lambda\tilde B\\
0&-\tilde B\Lambda&Q_u\end{bmatrix}\Bigg)\phi_\mathrm s\leq0.\end{multline}
A sufficinent condition for \eqref{eq:diss_condition_matrix_form} is
   \begin{equation*}
        \begin{bmatrix}
            P-Q_x \!\!&\!\! -\tilde{A}^\top\Lambda \!\!&\!\! 0 \\ -\Lambda\tilde{A} \!\!&\!\! 2\Lambda-\Lambda\tilde{B_s}-\tilde{B}_s^\top\Lambda \!\!&\!\! -\Lambda \tilde{B}\\
            0\!\!&\!\!-\tilde B^\top \Lambda \!\!&\!\! Q_u
        \end{bmatrix} - \begin{bmatrix}
            A^\top \\
            B_s^\top \\
            B^\top
        \end{bmatrix} P \begin{bmatrix}
            A \!\!&\!\! B_s \!\!&\!\! B
        \end{bmatrix}\succeq 0,
        \end{equation*}
which, by resorting to the Schur complement, is equivalent to~\eqref{eq:dISS_suff_cond}, completing the proof.
\hfill{}$\square$\smallskip\\
\textbf{Proof of Proposition~\ref{prop:identification_lmi}.}
The first step of the proof shows that solving~\eqref{opt:least_squares_identification} is equivalent to solving~\eqref{opt:identification_lmi}.
To show this, first we use~\eqref{eq:lmi_training_condition} in $J(\theta)=\cfrac{1}{N_{\mathrm d}^{\mathrm w}}\norm{Y_\mathrm d-\Phi\theta^\top}^2=\cfrac{1}{N_{\mathrm d}^{\mathrm w}}(Y_\mathrm d^\top Y_\mathrm d+\theta\Phi^\top\Phi\theta^\top-2\theta\Phi^\top Y_\mathrm d)$, where $N_{\mathrm d}^{\mathrm w}=N_\mathrm d-\tau_{\mathrm w}$
obtaining
\begin{align*}
&J(\theta)=\cfrac{1}{\gamma N_{\mathrm d}^{\mathrm w}}(\gamma Y_\mathrm d^\top Y_\mathrm d+\theta\gamma Q_\mathrm d\theta^\top-2\theta\gamma Q_\mathrm dQ_\mathrm d^{-1}\Phi^\top Y_\mathrm d)\\&\quad=\cfrac{1}{\gamma N_{\mathrm d}^{\mathrm w}}(\gamma Y_\mathrm d^\top Y_\mathrm d+\theta\tilde Q\tilde Q^{-1}\tilde Q\theta^\top-2\theta\tilde QQ_\mathrm d^{-1}\Phi^\top Y_\mathrm d).
\end{align*}
Now, we set $\theta=H\tilde Q^{-1}$, and we consider $H$ as an optimisation variable. It follows that 
\[J(H)=\cfrac{1}{\gamma N_{\mathrm d}^{\mathrm w}}(\gamma Y_\mathrm d^\top Y_\mathrm d+H\tilde Q^{-1}H^\top-2HQ_\mathrm d^{-1}\Phi^\top Y_\mathrm d).\]
Minimising $J(H)$ is equivalent to minimising
\[\tilde J(H)=(H-(Q_\mathrm d^{-1}\Phi^\top Y_\mathrm d)^\top\tilde Q)\tilde Q^{-1}(H^\top-\tilde Q(Q_\mathrm d^{-1}\Phi^\top Y_\mathrm d)).\]
The minimisation of $\tilde J(H)$ can be rewritten as
\begin{equation}
\begin{aligned}\label{eq:diss_opt_2}
    &\min_{s\in\R,H\in\R^{p\times n+m}}\, s\\
    &\text{subject to:}\\
    &\quad s\geq\tilde J(H)
\end{aligned}
\end{equation}
By resorting to the Schur complement, problem~\eqref{eq:diss_opt_2} is equivalent to \eqref{opt:identification_lmi}.\smallskip\\
The second step of the proof shows that condition~\eqref{eq:well_posedness_lmi} is equivalent to~\eqref{eq:well_posedness_ss}, i.e., it is a sufficient condition for the well-posedness of model~\eqref{eq:RNN_model_ss}. \\
Since $\tilde Q = \operatorname{diag}(\tilde  Q_{\mathrm e},Q_{s})$, it holds that $H_s = D_s Q_{s}$. Therefore, condition~\eqref{eq:well_posedness_lmi} can be rewritten as
\(
2Q_{s} - (\tilde B_s^{\mathrm{o}} + \tilde B_y D_s) Q_{s} - Q_{s}(\tilde B_s^{\mathrm{o}} + \tilde B_y D_s)^\top \succ 0 .
\)
Left- and right-multiplying this inequality by $\Lambda = Q_{s}^{-1}$, and recalling that $\tilde B_s = \tilde B_s^{\mathrm{o}} + \tilde B_y D_s$, yields condition~\eqref{eq:well_posedness_ss}, which completes the proof.
\hfill{}$\square$\smallskip\\
\textbf{Proof of Proposition~\ref{prop:contarctivity_implies_diss}.}
To prove \emph{(i)}, we define matrices $B_{\tilde u}\coloneq[B_u,\,B_y]$ and $\tilde B_{\tilde u}\coloneq[\tilde B_u,\,\tilde B_y]$, and the vector $\tilde u(k)\coloneq[u(k)^\top,\,y(k)^\top]^\top$.\\
Model~\eqref{eq:RNN_model_x}-\eqref{eq:RNN_model_s} can be rewritten as
\begin{equation}\label{eq:RNN_model_yu}
\begin{aligned}
&x(k+1) = A_x x(k) + B_{\tilde{u}} \tilde{u}(k) + B_s^{\mathrm{o}} s(k), \\
&s(k) = \sigma ( \tilde{A}_x x(k) + \tilde{B}_{\tilde{u}} \tilde{u}(k) + \tilde{B}_s^{\mathrm{o}} s(k)).
\end{aligned}
\end{equation}
In view of Lemma~\ref{th:dISS_condition}, a sufficient condition such that \eqref{eq:RNN_model_yu} is $\delta$ISS is that there exist matrices $P,Q_x\in\mathbb S_+^{n}$, $Q_{\tilde u} \in \mathbb S_+^{m+p}$, and $\Lambda \in\mathbb D_+^\nu$, such that 
\begin{equation}\label{eq:dISS_suff_cond_untraned_model}
\begin{bmatrix}
    P - Q_x 
    & -\tilde{A}_x^\top \Lambda 
    & 0 
    & A_x^\top P \\
    -\Lambda \tilde{A}_x 
    & 2\Lambda - \Lambda \tilde{B}_s^{\mathrm{o}} - \tilde{B}_s^{\mathrm{o}\,\top}\Lambda 
    & -\Lambda \tilde{B}_{\tilde{u}} 
    & {B}_s^{\mathrm{o}\,\top} P \\
    0 
    & -\tilde{B}_{\tilde{u}}^\top \Lambda 
    & Q_{\tilde{u}} 
    & B_{\tilde{u}}^\top P \\
    P A_x 
    & P B_s^{\mathrm{o}} 
    & P B_{\tilde{u}} 
    & P
\end{bmatrix}
\succeq 0.
\end{equation}
By congruence transformation~\cite{boyd1994linear}, condition~\eqref{eq:dISS_suff_cond_untraned_model} is equivalent to the condition
\begin{equation}\label{eq:dISS_suff_cond_untraned_model_congr}
\begin{bmatrix}
    P - Q_x 
    & -\tilde{A}_x^\top \Lambda 
    & A_x^\top P 
    & 0 \\
    -\Lambda \tilde{A}_x 
    & 2\Lambda - \Lambda \tilde{B}_s^{\mathrm{o}} - \tilde{B}_s^{\mathrm{o}\,\top}\Lambda 
    & B_s^{\mathrm{o}\,\top} P 
    & -\Lambda \tilde{B}_{\tilde{u}} \\
    P A_x 
    & P B_s^{\mathrm{o}} 
    & P 
    & P B_{\tilde{u}} \\
    0 
    & -\tilde{B}_{\tilde{u}}^\top \Lambda 
    & B_{\tilde{u}}^\top P 
    & Q_{\tilde u}
\end{bmatrix}
\succeq 0.
\end{equation}
We set $P=P_\mathrm o$, $\Lambda=\Lambda_\mathrm o$, and $Q_x=(1-\bar\alpha^{2})P_\mathrm o$. By leveraging the Schur complement, condition~\eqref{eq:dISS_suff_cond_untraned_model_congr} can be rewritten as
\begin{equation}\label{eq:suff_cond_contractivity_diss}
M_\mathrm P-H_\mathrm P^\top Q_{\tilde u}^{-1}H_\mathrm P\succeq0,
\end{equation}
where $H_\mathrm P=[0,\, -\tilde B_{\tilde u}^{\!\top}\Lambda_\mathrm{o},\, B_{\tilde u}^{\!\top} P_\mathrm{o}]$ and
\[M_\mathrm{P} =
\begin{bmatrix}
    \bar{\alpha}^2 P_\mathrm{o} 
    & -\tilde{A}_x^\top \Lambda_\mathrm{o} 
    & A_x^\top P_\mathrm{o} \\
    -\Lambda_\mathrm{o} \tilde{A}_x 
    & 2\Lambda_\mathrm{o} 
        - \Lambda_\mathrm{o} \tilde{B}_s^{\mathrm{o}} 
        - \tilde{B}_s^{\mathrm{o}\,\top}\Lambda_\mathrm{o} 
    & {B_s^{\mathrm{o}}}^\top P_\mathrm{o} \\
    P_\mathrm{o} A_x 
    & P_\mathrm{o} B_s^{\mathrm{o}} 
    & P_\mathrm{o}
\end{bmatrix}.\]
Since~\eqref{eq:RENcontractivity_Lyapunov} holds by assumption, it follows that
$M_\mathrm{P} \succeq \bar{\ell} I_{2n+\nu}$,
for some \(\bar{\ell} \in \mathbb{R}_{+}\).  
Therefore, condition~\eqref{eq:suff_cond_contractivity_diss} is satisfied if
\(
H_\mathrm{P}^\top Q_{\tilde u}^{-1} H_\mathrm{P} \preceq \bar{\ell} I_{2n+\nu},
\)
which holds if and only if
\begin{equation}\label{eq:contr_diss_final_condition}
\| H_\mathrm{P}^\top Q_{\tilde u}^{-1} H_\mathrm{P} \| \leq \bar{\ell}.
\end{equation}
Note that there always exists \(Q_{\tilde u}\in\mathbb S_+^{m+p}\) which ensures \eqref{eq:contr_diss_final_condition}.  
For example, noting that
\(
\|H_\mathrm{P}^\top Q_{\tilde u}^{-1} H_\mathrm{P}\| 
    \le \|H_\mathrm{P}\|^{2} \, \|Q_{\tilde u}^{-1}\|\leq\frac{\|H_\mathrm{P}\|^{2}}{\lambda_{\mathrm min}(Q_{\tilde u})},
\)
we can choose
$Q_{\tilde u} = \diag(\gamma_1,\dots,\gamma_{m+p})$, with $\gamma_i \ge \|H_\mathrm{P}\|^{2}/\bar{\ell}$ for all $i=1,\dots,m+p$.\smallskip\\
The proof of \emph{(ii)} follows the same arguments of the proof of~\cite[Proposition~3.1]{d2025data}.\\  Defining $\tilde u_\mathrm d(k)\coloneq[u_\mathrm d(k)^\top,\,y_\mathrm d(k)^\top]^\top$, the $\delta$ISS property of~\eqref{eq:RNN_model_yu} implies that there exist functions $\beta\in\mathcal{KL}$ and $\gamma\in\mathcal{K}_\infty$ such that for any $k\in\mathbb Z_+$, it holds that
\(\norm{x_\mathrm d(k)-x(k)}\leq\beta(\norm{x_\mathrm d(0)-x(0)},k)+\gamma(\max_{h\geq0}\norm{\tilde u_\mathrm d(h)-\tilde u(h)}).\)
Since
\(\tilde u_\mathrm d(k)-\tilde u(k)=[
    0,\,\eta(k)^\top]
^\top,\)
and $\eta(k)\leq\bar\eta$ by assumption,
it follows that $\max_{h\geq0}\norm{\tilde u_\mathrm d(h)-\tilde u(h)}\leq\bar\eta$, and therefore that
\(\norm{x_\mathrm d(k)-x(k)}\leq\bar w_x(k)\).\smallskip\\
To prove \emph{(iii)}, we define $\Delta x\coloneq x_\mathrm d(k) -x(k)$,  $\Delta s\coloneq s_\mathrm d(k) -s(k)$, $v_\mathrm d\coloneq\tilde A_x x_\mathrm d(k)+\tilde B_uu_\mathrm d(k)+\tilde B_s^\mathrm os_\mathrm d(k)+\tilde B_yy_\mathrm d(k)$, and $v\coloneq\tilde A_x x(k)+\tilde B_uu(k)+\tilde B_s^\mathrm os(k)+\tilde B_y y(k)$.
It follows that
\begin{equation}\label{eq:dv_evolution}
    \Delta v\coloneq v_\mathrm d-v=\tilde A_x\Delta x+\tilde B_s^\mathrm o\Delta s+\tilde B_y\eta(k).
\end{equation}
By the mean value theorem, since $\sigma^{(i)}(\cdot)$ is continuous and differentiable, for all $v^{(i)},v^{(i)}+\Delta v^{(i)}\in\R$, there exists a scalar $v_i^\star$ such that $v_i^\star\in[v^{(i)},\,v^{(i)}+\Delta v^{(i)}]$ if $\Delta v^{(i)}\geq 0$, or $v_i^\star\in[v^{(i)}+\Delta v^{(i)},\,v^{(i)}]$ if $\Delta v^{(i)}<0$, such that 
\begin{equation}\label{eq:delta_s_single_component}
\begin{aligned}
    \Delta s^{(i)}&\coloneq \sigma^{(i)}(v^{(i)}+\Delta v^{(i)})-\sigma(v^{(i)})\\
    &=\delta\sigma_i(v_i^\star)\Delta v^{(i)},\quad \forall i=1,\dots,\nu
    \end{aligned}
\end{equation}
where $\delta\sigma_i(v_i^\star)\coloneq\frac{\partial\sigma^{(i)}(v^{(i)})}{\partial v^{(i)}}\Bigg|_{v_i^\star}$.\\
Defining $\Sigma(k)\coloneq\diag(\delta\sigma_1(v_1^\star),\dots,\delta\sigma_\nu(v_\nu^\star))$ and using~\eqref{eq:dv_evolution},
we can write~\eqref{eq:delta_s_single_component} in compact form as
\begin{equation}\label{eq:deltas_SM}
\Delta s=\Sigma(k)\Delta v=\Sigma(k)(\tilde A_x\Delta x+\tilde B_s^\mathrm o\Delta s+\tilde B_y\eta(k)).
\end{equation}
In view of Assumption~\ref{ass:sigmoid_function}, it holds that $0<\delta\sigma_i(v_i^\star)\leq 1$, for all $i=1,\dots,\nu$, which implies $\Sigma(k)\in\mathbb D_+^\nu$ and $\Sigma(k)\preceq I_\nu$.\\
Since \eqref{eq:RNN_model} satisfies condition~\eqref{eq:RENcontractivity_Lyapunov} by assumption, which implies $M_\mathrm P\succ 0$, it follows that there exist $\Lambda_\mathrm o\in\mathbb D_+$ such that 
\begin{equation}\label{eq:wel_posedness_untrained_model}
    2\Lambda_\mathrm{o} - \Lambda_\mathrm{o}\tilde{B}_s^{\mathrm{o}} - \tilde{B}_s^{\mathrm{o}\,\top}\Lambda_\mathrm{o}\succ0.
\end{equation}
According to~\cite[Lemma~1]{ravasio2025developmentvelocityformclass}, condition~\eqref{eq:wel_posedness_untrained_model} implies that the matrix $I_\nu-\Sigma(k)\tilde B_s^{\mathrm o}$ is full rank and hence invertible. Consequently, solving \eqref{eq:deltas_SM} for $\Delta s$ yields
\(\Delta s=(I_\nu-\Sigma(k)\tilde B_s^\mathrm o)^{-1}\Sigma(k)(\tilde A_x\Delta x+\tilde B_y\eta(k)).\)
Taking the norm of $\Delta s$, it follows that
$\norm{\Delta s}\leq\norm{(I_\nu-\Sigma(k)\tilde B_s^\mathrm o)^{-1}\Sigma(k)\tilde A_x}\norm{\Delta x}+\norm{(I_\nu-\Sigma(k)\tilde B_s^\mathrm o)^{-1}\Sigma(k)\tilde B_y}\norm{\eta(k)}.$
Finally, exploiting the bounds $\|\Delta x\|\leq \bar w_x(k)$ and $\|\eta(k)\|\leq \bar\eta$, it follows that $\|\Delta s\|\leq \bar w_s(k)$, completing the proof.
\hfill{}$\square$\smallskip\\
\textbf{Proof of Proposition~\ref{prop:diss_opt}.}
To prove Proposition~\ref{prop:diss_opt}, we need to show that condition~\eqref{eq:diss_lmi} implies~\eqref{eq:dISS_suff_cond}, and hence it guarantees the $\delta$ISS of \eqref{eq:RNN_model_ss}.\\
Since $\tilde Q = \operatorname{diag}(Q_\mathrm C,Q_\mathrm D.Q_{s})$ and $H=\theta^\star\tilde Q$, it follows that $H_x = C Q_\mathrm C$, $H_u = D Q_\mathrm D$, and $H_s = D_s Q_{s}$. 
Moreover, recalling that $A=A_x+B_yC$, $B=B_u+B_yD$, $B_s=B_s^\mathrm o+B_yD_s$, $\tilde A=\tilde A_x+\tilde B_yC$, $\tilde B=\tilde B_u+\tilde B_yD$, and $\tilde B_s=\tilde B_s^\mathrm o+\tilde B_yD_s$, condition~\eqref{eq:diss_lmi} is equivalent to
\begin{equation}\label{eq:diss_lmi_2}
            \begin{bmatrix}
            Q_\mathrm C-\tilde Q_x & -Q_\mathrm C\tilde{A}^\top & 0 &Q_\mathrm CA^\top\\ -\tilde{A}Q_\mathrm C & 2Q_{s} - \tilde B_s Q_{s} - Q_{s}\tilde B_s^\top & -\tilde{B}Q_\mathrm D & Q_{s}B_s\\
            0&-Q_\mathrm D\tilde B^\top& \tilde Q_u & Q_\mathrm DB^\top\\
            AQ_\mathrm C & B_sQ_{s} & BQ_\mathrm D & Q_\mathrm C\end{bmatrix}\succeq0.
\end{equation}
Now, set $P= Q_\mathrm C^{-1}$, $Q_x=P\tilde Q_xP$, $Q_u=Q_\mathrm D^{-1}\tilde Q_uQ_\mathrm D^{-1}$, and $\Lambda=Q_{s}^{-1}$. Left- and right-multiplying~\eqref{eq:diss_lmi_2} by $M=\diag(P,\Lambda,Q_\mathrm D^{-1},P)$ and $M^\top$, we obtain~\eqref{eq:dISS_suff_cond}, completing the proof.
\hfill{}$\square$\smallskip\\
\textbf{Proof of Proposition~\ref{prop:diss_submodel}}
To prove Proposition~\ref{prop:diss_submodel}, we rewrite the dynamics~\eqref{eq:PI_RNN_subsystem_i} with respect to the input $v_i$ as
\begin{equation}\label{eq:PI_RNN_subsystem_i_v}
    \begin{aligned}
        &x_i(k+1){=}A_{x,i}x_i(k)+B_{v,i}^\mathrm ov_i(k){+}B_{s,i}^\mathrm os_i(k){+}B_{y,i}y_i(k)\\
        &s_i(k)=\sigma_i(\tilde A_{x,i}x_i(k)+\tilde B_{v,i}^\mathrm ov_i(k){+}\tilde B_{s,i}^\mathrm os_i(k){+}\tilde B_{y,i}y_i(k))\\
        &y_i(k)= C_{i,i}x_i(k){+} D_{v,i}v_i(k){+}D_{s,i}s_i(k)
    \end{aligned}
\end{equation}
Defining $B_{v,i}\coloneq B_{v,i}^\mathrm o+B_{y,i}D_{v,i}$ and $\tilde B_{v,i}\coloneq \tilde B_{v,i}^\mathrm o+\tilde B_{y,i}D_{v,i}$, \eqref{eq:PI_RNN_subsystem_i_v} can be rewritten in state space form as
\begin{equation}\label{eq:RNN_submodel_ss_v}
    \begin{aligned}
        &x_i(k+1)=A_ix_i(k)+B_{v,i}v_i(k)+B_{s,i}s_i(k)\\
        &s_i(k)=\sigma_i(\tilde A_ix_i(k)+\tilde B_{v,i}v_i(k)+\tilde B_{s,i}s_i(k))\\
        &y_i(k)= C_{i,i}x_i(k)+ D_{v,i}v_i(k)+D_{s,i}s_i(k)
    \end{aligned}
\end{equation}
Since the dynamics~\eqref{eq:RNN_submodel_ss_v} have the same structure as~\eqref{eq:RNN_model_ss}, the result in Proposition~\ref{prop:diss_submodel} follows from Proposition~\ref{prop:diss_opt}.
\hfill$\square$\smallskip\\
\textbf{Proof of Proposition~\ref{prop:diss_non_overlapping}}
In view of \eqref{eq_cstr_diss_non_overl_H}-\eqref{eq_cstr_diss_non_overl_Q}, it follows that for all $i\in\mathcal I$, $(H_x)_{\{i\},\{j\}}=C_{i,j}(Q_\mathrm C)_{\mathcal N_{x,i}^\mathcal M\cup\{i\},\mathcal N_{x,i}^\mathcal M\cup\{i\}}$ if $j\in\mathcal N_{x,i}^\mathcal M\cup\{i\}$ and  $(H_x)_{\{i\},\{j\}}=C_{i,j}=0$ if $j\notin\mathcal N_{x,i}^\mathcal M\cup\{i\}$. Therefore, it follows that $H_x = C Q_{\mathrm C}$. By applying a similar reasoning, it is possible to show that $H_u = D Q_{\mathrm D}$, and $H_s = D_s Q_{s}$. Therefore, assuming that~\eqref{eq:diss_lmi} holds, $\mathcal M$ is $\delta$ISS in view of Proposition~\ref{prop:diss_opt}.
\hfill$\square$


\bibliographystyle{cas-model2-names}

\bibliography{cas-refs}



\end{document}